\documentclass[10pt,letterpaper]{article}
\usepackage[top=0.85in,left=2.75in,footskip=0.75in]{geometry}

\usepackage{amsmath,amssymb}
\usepackage{changepage}
\usepackage[utf8x]{inputenc}
\usepackage{textcomp,marvosym}
\usepackage{cite}
\usepackage{nameref,hyperref}
\usepackage{microtype}
\DisableLigatures[f]{encoding = *, family = * }
\usepackage{amsthm}

\usepackage[table]{xcolor}
\usepackage{array}
\newcolumntype{+}{!{\vrule width 2pt}}
\newlength\savedwidth

\theoremstyle{definition}

\newtheorem*{thm}{Theorem}
\newtheorem*{lem}{Lemma}

\usepackage[aboveskip=1pt,labelfont=bf,labelsep=period,justification=raggedright,singlelinecheck=off]{caption}

\makeatletter
\renewcommand{\@biblabel}[1]{\quad#1.}
\makeatother

\date{}

\usepackage{lastpage,fancyhdr,graphicx}
\usepackage{epstopdf}

\begin{document}
\baselineskip=16pt

\begin{flushleft}
{\Large
\textbf\newline{Extracting replicable associations across multiple studies: algorithms for controlling the false discovery rate} 
}
\newline
\\
David Amar\textsuperscript{1},
Ron Shamir\textsuperscript{1*},
Daniel Yekutieli\textsuperscript{2},
\\
\bigskip
\textbf{1} The Blavatnik School of Computer Science, Tel Aviv University, Tel Aviv, Israel
\\
\textbf{2} Department of Statistics and OR, Tel Aviv University, Tel Aviv, Israel
\\
\bigskip

* rshamir@tau.ac.il

\end{flushleft}
\section*{Abstract}
In almost every field in genomics, large-scale biomedical datasets are used to report associations. Extracting associations that recur across multiple studies while controlling the false discovery rate is a fundamental challenge. Here, we consider an extension of Efron's single-study two-groups model to allow joint analysis of multiple studies. We assume that given a set of p-values obtained from each study, the researcher is interested in associations that recur in at least $k>1$ studies. We propose new algorithms that differ in how the study dependencies are modeled. We compared our new methods and others using various simulated scenarios. The top performing algorithm, SCREEN (Scalable Cluster-based REplicability ENhancement), is our new algorithm that is based on three stages: (1) clustering an estimated correlation network of the studies, (2) learning replicability (e.g., of genes) within clusters, and (3) merging the results across the clusters using dynamic programming. 

We applied SCREEN to two real datasets and demonstrated that it greatly outperforms the results obtained via standard meta-analysis. First, on a collection of 29 case-control large-scale gene expression cancer studies, we detected a large up-regulated module of genes related to proliferation and cell cycle regulation. These genes are both consistently up-regulated across many cancer studies, and are well connected in known gene networks. Second, on a recent pan-cancer study that examined the expression profiles of patients with or without mutations in the HLA complex, we detected an active module of up-regulated genes that are related to immune responses. Thanks to our ability to quantify the false discovery rate, we detected thrice more genes as compared to the original study. Our module contains most of the genes reported in the original study, and many new ones. Interestingly, the newly discovered genes are needed to establish the connectivity of the module.

\section*{Lay Summary}

When analyzing results from multiple studies, extracting replicated associations is the first step towards making new discoveries. The standard approach for this task is to use meta-analysis methods, which usually make an underlying null hypothesis that a gene has no effect in all studies. On the other hand, in replicability analysis we explicitly require that the gene will manifest a recurring pattern of effects. In this study we develop new algorithms for replicability analysis that are both scalable (i.e., can handle many studies) and allow controlling the false discovery rate. We show that our main algorithm called SCREEN (Scalable Cluster-based REplicability ENhancement) outperforms the other methods in simulated scenarios. Moreover, when applied to real datasets, SCREEN greatly extended the results of the meta-analysis, and can even facilitate detection of new biological results.


\section*{Introduction}

Confidence in reported findings is a prerequisite for advancing any scientific field. Such confidence is achieved by showing replication of discoveries by further evidence from new studies \cite{McNutt2014}. In recent years, a new type of methodology called \textit{replicability analysis}, sometimes referred to as reproducibility analysis, was suggested as a way to statistically quantify the replication of discoveries across studies while controlling for the false discovery rate (FDR) \cite{Benjamini2009}. This type of analysis is crucial in studies that aim to detect new hypotheses by integrating existing data from multiple high-throughput experiments. 

The practical importance of replicability analysis is twofold. First, it is a tool for quantifying replication, and therefore the reliability, of reported results. This is of vital importance as in recent years concerns have been raised in several domains regarding low reproducibility, including economics \cite{Camerer2016}, psychology \cite{Braver2014}, medicine \cite{Ioannidis2005}, and biological studies that rely on high throughput experiments such as gene expression profiling \cite{Wirapati2008,Laas2016a}, and network biology \cite{Verleyen2016}. Second, collating information from multiple studies can lead to scientific results that may be beyond the reach of a single study. Indeed, replicability analysis was demonstrated as a tool for extracting new results by merging Genome Wide Association Studies (GWAS) \cite{HellerYekutieli2014}.

The underlying assumption in standard meta-analysis is that the multiple studies estimate the same effect. Aggregating information across studies produces estimators with smaller measurement error that yield considerably more power to reject the null hypothesis regarding this effect. While meta-analyses are widely applied and have been extensively studied in the statistical literature \cite{Hedges1985} and in computational biology \cite{Chang2013,Li2014}, in recent years the changes in the scale and also the scope of public high-throughput biomedical data has led to new methodological challenges. For example, Zeggini et al. \cite{Zeggini2007} analyzed results of genome-wide association scans for Type 2 Diabetes (T2D) on the same set of almost $2.5$ million SNPs from eight study populations. In such situations, the first, and more obvious, challenge is accounting for inflation in the number of false discoveries due to the multiplicity of outcomes. The second challenge is hidden in the null hypothesis that the effect size is $0$ in all the studies (as done in meta-analysis). That assumption is oblivious to the consistency of the effects, and thus it overlooks important scientific information. Third, there is a need to distinguish between true effects that are specific to a single study and true effects that represent general discoveries that are replicable. For example, Kraft et al. \cite{Kraft2009} suggested that for common genetic variants, the anticipated effects on the phenotype may be very similar to population biases in individual genetic association studies. While these are real discoveries in the sense that similar estimated effects are expected to be observed if the experiment could be replicated,  their scientific importance is limited because they are specific to a particular study population. For this reason, the authors argue that it is important to see the association in additional studies conducted using a similar, but not identical, study base.

In recent years several frequentist approaches were suggested for replicability analysis. Benjamini and Heller \cite{Benjamini2008} introduced an inferential framework for replicability that is based on tests of partial conjunction null hypotheses. For a meta-analysis of $n$ studies of the same $m$ outcomes and $u = 1 \ldots n$, the partial conjunction $H^{u/n} (g)$ is that outcome $g$ has a non-null effect in less than $u$ studies. Thus $H^{1/n} (g)$ is the standard meta-analysis null hypothesis that outcome $g$ has a null effect in all $n$ studies. The authors also introduced p-values for testing $H^{u/n} (g)$ for each outcome. Benjamini, Heller and Yekutieli \cite{Benjamini2009} applied the Benjamini-Hochberg FDR procedure \cite{Benjamini1995} (BH) to the partial conjunction hypotheses p-values, and suggested setting  $u = 2$ in order to assess replicability. Heller et al. \cite{Heller2014} developed an approach for checking if a follow-up study corroborates the results reported in the original study. Song and Tseng \cite{Song2014} proposed a method to evaluate the proportion of non-null effects of a gene. However, they used a standard meta-analysis null hypothesis and their method cannot handle a complementary form between the null hypothesis and the alternative hypothesis (i.e., composite hypothesis). As we discuss below, Bayesian methods handle these shortcomings and naturally offer a powerful framework for replicability analysis.

Heller and Yekutieli \cite{HellerYekutieli2014} introduced repfdr: an extension of the single-study empirical Bayes approach of Efron \cite{Efron2010} for testing the partial conjunction hypotheses in the multi-study case. They estimate the posterior probabilities of the various configurations of outcome effect status (null or non-null) across studies, and compute the local Bayes FDRs for each partial conjunction null by summing the posterior probabilities for the relevant configurations. The authors showed that their approach controls the FDR and offers more power as compared to the frequentist methods. Moreover, the advantage of their approach in power greatly increases for $u > 1$. However, the EM-based algorithm of the method is not scalable as the number of estimated parameters is exponential in the number of studies.

In this study, we developed three new empirical Bayes methods for FDR-controlled replicability analysis of many studies. The input data are a matrix of p-values (or z-scores), where rows represent the measured objects (e.g., genes) and columns represent the studies. Our three methods differ in the way the column dependency is modeled. First, under independence assumption we estimate replicability using dynamic programming. Second, if no assumptions are made, we propose a scalable extension to the EM algorithm of repfdr \cite{HellerYekutieli2014}. Finally, for the case where the studies are assumed to originate from independent clusters we present a new algorithm called SCREEN (Scalable Cluster-based REplicability ENhancement). We compared these methods and others that are not based on FDR estimation using various simulated scenarios and showed that SCREEN was consistently among the top performing algorithms.

We applied SCREEN to two cancer datasets, where each is a collection of multiple case-control gene expression experiments. In both cases SCREEN greatly improved the results obtained by standard meta-analysis, to a point where new biological insights emerge. The first dataset is a collection of 29 case-control gene expression cancer studies from different tissues. Here, SCREEN detected a large set of genes that are consistently up-regulated, highly enriched for cell proliferation and cell cycle regulation functions, and are well connected in known gene networks, indicating their functional coherence. The second dataset is a recent pan-cancer study that examined the expression profiles of patients with or without mutations in the HLA complex across 11 cancer types \cite{Shukla2015}. Here, SCREEN detected a large set of  up-regulated genes that are related to immune responses. While this result was also detected in the original study we reported many more genes thanks to our ability to quantify the false discovery rate, and we detected prominent genes and enriched pathways that were not reported previously. 


\section*{Results}

\subsection*{Preliminaries and notations}
We start with a brief introduction to the single-study model. For a full description, theoretical justification, and relation to the classic BH method, see \cite{EfronBook2010}. Given a large set of $N$ hypotheses tested in a large-scale study, the \textit{two-groups model} provides a simple Bayesian framework for multiple testing: each of the $N$ cases (e.g., genes in a gene expression study) are either null or non-null with prior probability $\pi_0$ and $\pi_1 = 1 - \pi_0$, and with z-scores (or p-values) having density either $f_0(z)$ or $f_1(z)$. Also, when the assumptions of the statistical test are valid, we know that the $f_0$ distribution is a standard normal (or a uniform distribution for p-values), and we call it the \textit{theoretical null}. The mixture density and probability distributions are:
\begin{eqnarray*}
f(z) = \pi_0 f_0(z) + \pi_1 f_1(z)\\
F(z) =\pi_0 F_0(z) + \pi_1 F_1(z)
\end{eqnarray*} 
For a rejection area $\mathcal{Z}_y = (-\infty,y)$, using Bayes rule we get:
\begin{eqnarray*}
Fdr(\mathcal{Z}_y) \equiv Pr\{ null| z \in \mathcal{Z}_y\} = \pi_0F_0(y)/F(y)
\end{eqnarray*}
We call $Fdr$ the \textit{(Bayes) false discovery rate} for $\mathcal{Z}$: this is the probability we would make a false discovery if we report $\mathcal{Z}$ as non-null. If $\mathcal{Z}$ is a single point $z_0$ we define the \textit{local (Bayes) false discovery rate} as:
\begin{eqnarray*}
fdr(z_0) \equiv Pr\{ null | z = z_0 \} = \pi_0f_0(z_0) / f(z_0)
\end{eqnarray*} 

In this work we consider an extension of Efron's two group model to analyze data of $n$ genes over $m>1$ studies. The data for gene $i$ are a vector of m statistics $Z_{i,\cdot} = (Z_{i,1}, \cdots, Z_{i,m})$ that are all either z-scores or p-values. For simplicity, from now on we assume that these data are z-scores. The unknown parameter for gene $i = 1,\cdots, n$ is a binary configuration vector $H_{i,\cdot} = (H_{i,1}, \cdots, H_{i,m})$, with $H_{i,j} \in \{0,1\}$. If $H_{i,j} = 0$ then gene $i$ is a $null$ realization in study $j$, and it is a $non-null$ realization otherwise. 

We assume that in each study $j$ the parameters of the two-groups model $\theta_j$: $(\pi_0^j, f_0^j, f_1^j, f^j)$ are fixed and focus on replicability analysis. In this study, we tested two methods for two-groups estimation: \textit{locfdr} \cite{Efron2004} and estimation based on mixture of Gaussians that we call \textit{normix}, see \textbf{Materials and Methods} for a full description.

Generally, unless mentioned otherwise, we also assume that the genes are independent. However, note that estimation of $\theta_j$ could account for gene dependence within single studies \cite{Efron2007,Efron2009}. Finally, we also assume that the z-scores of a gene are independent given its configuration. That is,
$$P(Z_{i,\cdot}|H_{i,\cdot}) = \prod_{j=1}^m P(Z_{i,j}|H_{i,j}) = \prod_{j=1}^m \left( f_0^j(Z_{i,j}) \right)^{(1-H_{i,j})} \left( f_1^j(Z_{i,j}) \right)^{H_{i,j}}$$

Next, we use $h\in\{0,1\}^m$ to denote an arbitrary configuration vector, and $\pi(h)$ to denote a probability assigned to the parameter space. We assume that the researcher has a set of configurations ${\cal{H}}_1 \subseteq \{0,1\}^m$ that represents the desired rejected genes. Specifically, in this work we assume that the researcher is interested in genes that are non-null in at least $k$ studies: ${\cal{H}}_1 = \{h:|h| \ge k\}$, where $|h| = \sum_{j=1}^m h_j$. 

As a note, selection of $k$ depends on the research question at hand. For example, Heller and Yekutieli used  $k=2$ to detect a minimal replicability of SNPs in a GWAS \cite{HellerYekutieli2014}. Such low $k$ values can also be reasonable if the $m$ studies represent different biological questions that are related, such as differential expression experiments from different cancer subtypes. On the other hand, if the $m$ studies represent tightly related experiments such as biological replicates then larger k (e.g., $m/2$) seems more reasonable.

The \textit{local false discovery rate} (fdr) of a gene $i$ can be formulated as:
$$
fdr(Z_{i,\cdot}) = Pr(\bar{{\cal{H}}_1} |Z_{i,\cdot}) = \sum_{h:h \not \in {\cal{H}}_1} P(h|Z_{i,\cdot}) = \sum_{h:h \not \in {\cal{H}}_1} \frac{P(Z_{i,\cdot}|h)P(h)}{P(Z_{i,\cdot})}
$$

For a given $k$ and ${\cal{H}}_1 = \{h:|h| \ge k\}$ we get: 
$$fdr_k (Z_{i,\cdot}) = \sum_{h:|h| < k} \frac{P(Z_{i,\cdot}|h)P(h)}{P(Z_{i,\cdot})}$$

In the next sections we present new fast algorithms for computing either $fdr_k$ exactly or an upper bound of it. 

\subsection*{An $O(mnk)$ algorithm under independence assumption}
\begin {lem}If the studies are independent (in the parameter space) then:
$$
fdr_k (Z_{i,\cdot}) = \sum_{h:|h| < k} \prod_{j=1}^m \frac{P(Z_{i,j}|h_j)P(h_j)}{f^j(Z_{i,j})}
$$
\end{lem}
\begin{proof}
First, note that under the independence assumption $P(h) = \prod_{j=1}^m P(h_j) =  \prod_{j=1}^m {\pi_0^j}^{(1-h_j)}(1-\pi_0^j)^{h_j}$. Second, as the z-scores are independent given the configuration vector $h$ we get that:
\begin{eqnarray*}
P(Z_{i,\cdot}) &=& \sum_h P(Z_{i,\cdot}|h)P(h) = \sum_h \prod_{j=1}^m P(Z_{i,j}|h_j) {\pi_0^j}^{(1-h_j)}(1-\pi_0^j)^{h_j}\\
&=& \prod_{j=1}^m \left ( P(Z_{i,j}|h_j=0) \pi_0^j + P(Z_{i,j}|h_j=1)(1-\pi_0^j) \right )
\end{eqnarray*} 
\end{proof}

\textbf{Proposition:} If the studies are independent then $fdr_k$ can be computed in $O(mnk)$.

\begin{proof} By the lemma, the $fdr$ of a gene is based on the product of the two-group model densities in each study. Therefore:
$$
fdr_k^{indep} (Z_{i,\cdot}) = \sum_{h:|h| < k} \ \prod_{j=1}^m \frac{(\pi_0^j f_0^j (Z_{i,j}))^{1-h_j}((1-\pi_0^j) f_1^j (Z_{i,j}))^{h_j}}{f^j(Z_{i,j})}
$$

We use dynamic programming to calculate $fdr_k(z_i)$ as follows. Define:
$$
U[i,j,k^*] = \sum_{h:|h| = (k^*-1)} \ \prod_{j=1}^m \frac{(\pi_0^j f_0^j (Z_{i,j}))^{1-h_j}((1-\pi_0^j) f_1^j (Z_{i,j}))^{h_j}}{f^j(Z_{i,j})}
$$
These values can be calculated (for each gene $i$) by updating a table of $m \times (k+1)$ values. The base cases are:
$$U[i,j,1] = \prod_{j=1}^m \frac{\pi_0^j f_0^j (Z_{i,j})}{f^j(Z_{i,j})}$$

The recursive formulas are:
\begin{eqnarray*}
U[i,j,k^*] &=&  \frac{\pi_0^j f_0^j (Z_{i,j})}{f^j(Z_{i,j})} U[i,j-1,k^*] + \frac{(1-\pi_0)^j f_1^j (Z_{i,j})}{f^j(Z_{i,j})} U[i,j-1,k^*-1]
\end{eqnarray*} 

Finally, to obtain the $fdr$ of a gene we sum over the values in the last column:
$$
fdr_k^{indep} (Z_{i,\cdot}) = \sum_{k^*=1}^{k-1} U[i,m,k^*]
$$

The running time for analyzing each gene is $O(mk)$ and the total running time is $O(nmk)$.
\end{proof} 

\subsection*{Schemes for handling dependence}
\subsubsection*{Approximating the prior using restricted EM}
The empirical Bayes method of \cite{HellerYekutieli2014} estimates prior distribution $\pi(h)$ directly from the data. However, this approach has two drawbacks. First, the EM algorithm explicitly keeps a value for each possible configuration, which makes the algorithm intractable when $m$ increases. Second, the estimation for rare configurations might be inaccurate, unless $n>>2^m$.

As an alternative, we develop an algorithm that keeps in memory only a small set of high probability configurations. We then use these estimates to obtain an upper bound for the $fdr$ of a gene. We first describe the EM without any constraints on the configuration space, and then show that the same process can be used to obtain an estimator in the constrained case. That is, the EM is guaranteed to improve the solution and converge.
The EM formulation is based on repfdr \cite{HellerYekutieli2014,Yekutieli2014}, as follows: 

The E-step:
$$
P(H_{i,\cdot} = h | Z_{i,\cdot}, \pi^{(t)}(h)) = \frac{f(Z_{i,\cdot}|h)\pi^{(t)}(h)}{\sum_{h'} f(Z_{i,\cdot}|h')\pi^{(t)}(h')}
$$

The M-step:
$$
\pi^{(t+1)}(h) = \frac{1}{n} \sum_i P(H_{i,\cdot} = h | Z_{i,\cdot}, \pi^{(t)}(h)) =  \frac{1}{n} \sum_i \frac{f(Z_{i,\cdot}|h)\pi^{(t)}(h)}{\sum_{h'} f(Z_{i,\cdot}|h')\pi^{(t)}(h')}
$$
This process guarantees convergence to a local optimum. Our goal is to limit the search space. 
\begin{lem}
The EM algorithm above can be used to find a local optimum estimator under the constraint $\forall h \not \in {\cal{H}}' \ \pi(h)=0$, for any non-empty configuration set ${\cal{H}}'$. 
\end{lem}
\begin{proof}
Note that during the EM iterations, if at some time point $t$ $\pi^{(t)}(h)=0$ then $\forall t*>t $ $\pi^{(t*)}(h)=0$. Therefore, setting the starting point of the EM such that 
 $\forall  h \not \in {\cal{H}}' $ $\pi^{(0)}(h)=0$ satisfies the constraint and ensures convergence.
\end{proof}

\subsubsection*{Estimation under space complexity constraints}

We begin this section with additional notation. Given a configuration vector $h\in\{0,1\}^m$, let $h[l]$ be the vector containing the first $l$ entries of $h$. Given a real valued vector $v$, let $v_{(i)}$ denote the i'th smallest element of $v$.

Our algorithm is based on the simple observation that if at some point $\pi(h[l]) \le \epsilon$ then any extension of $h[l]$ cannot exceed $\epsilon$. That is, $\pi(h[l+1]) \le \epsilon$ regardless of the new value in position $l+1$. Our algorithm works as follows. The user specifies a limit to the number configurations kept in the memory - $n_H$. For simplicity we assume that $n_H$ is a power of $2$. We first run the unrestricted EM algorithm on the first $log_2(n_H)-1$ studies. We then iteratively add a new study. In each iteration $l$ we keep four parameters: (1) $\hat{H}^l$ - the set of the top $n_H$ probability configurations, (2) $\hat{\pi}^l$ - the vector of their assigned probabilities, (3) $\hat{\xi}^l$ - an estimation of $\sum_{h[l] \in \hat{H}^l} \pi^{l}(h[l])$, and (4) $\hat{\epsilon}^l$ - an estimation of the maximal probability among the excluded configurations.

Initially, $l=log_2(n_H)-1$ and $\hat{H}^l$ contains all possible configurations of the first $l$ studies. In addition $\hat{\xi}^l=1$, and $\hat{\epsilon}^l=0$. In iteration $l+1$ we run the restricted EM algorithm on all possible extensions of $\hat{H}^l$. That is, the input configuration set for the EM is a result of adding either $1$ or $0$ at the $l+1$ position of each configuration in $\hat{H}^l$. The EM run produces initial estimations for our parameters, on which the following ordered updates are applied:
\begin{eqnarray}
\hat{\pi}^{l+1} &=&  \hat{\xi}^l \hat{\pi}^{l+1} \\
\hat{H}^{l+1} &=&  \{ h[l+1]; \hat{\pi}^{l+1}(h[l+1]) \ge \hat{\pi}^{l+1}_{(n_H / 2)} \} \\
\hat{\xi}^{l+1} &=& \sum_{h[l+1] \in \hat{H}^{l+1}} \pi^{l+1}(h[l+1])\\
\hat{\epsilon}^{l+1} &=& \max(\hat{\epsilon}^l,\max_{h[l+1] \not\in \hat{H}^{l+1}} (\hat{\pi}^{l+1}(h[l+1])))
\end{eqnarray}
Note that in step (2) above we keep the top $n_H/2$ configurations in $\hat{H}^{l+1}$. This set is then used as input to the EM run in the next iteration. We repeat the process above until $l=m$. The output of the algorithm is $\hat{H}^m, \hat{\pi}^m, \hat{\xi}^m, \hat{\epsilon}^m$. 

\subsubsection*{A fast algorithm for an upper bound for the fdr}
We now use the output of the algorithm from the previous section to obtain an upper bound for the $fdr_k$ of a gene in running time $O(m(k+ nH))$ for each gene. 

\begin{thm}
Given the prior probability $\pi(h)$ of each configuration $h$ in $\cal{H'} \subseteq \cal{H}$, and an upper bound $\epsilon$ for the probability of all excluded configurations, the following inequality holds:
$$
fdr_k(Z_{i,\cdot})  \le \frac{\sum_{h:|h| < k \land h \in \cal{H'}} P(Z_{i,\cdot}|h)(\pi(h) - \epsilon) + \epsilon \sum_{h:|h| < k } P(Z_{i,\cdot}|h)}{\sum_{h \in \cal{H'} } P(Z_{i,\cdot}|h)\pi(h)}
$$
\end{thm}

\begin{proof}
Given that for each $h \not \in \cal{H'}$ $\pi(h) \le \epsilon$, we get:
$$ 
fdr_k(Z_{i,\cdot}) = \sum_{h:|h| < k \land h \in \cal{H'}} P(h|Z_{i,\cdot})  +  \sum_{h:|h| < k \land h \not{\in} \cal{H'}} P(h|Z_{i,\cdot}) \le
$$ $$
\sum_{h:|h| < k \land h \in \cal{H'}} P(h|Z_{i,\cdot})  +  \epsilon \left (\sum_{h:|h| < k \land h \not{\in} \cal{H'}} \frac{P(Z_{i,\cdot}|h)}{P(Z_{i,\cdot})} \right)
$$

Thus:

$$ 
fdr_k(Z_{i,\cdot})  \le \sum_{h:|h| < k \land h \in \cal{H'}} \frac{P(Z_{i,\cdot}|h)\pi(h)}{P(Z_{i,\cdot})}  +  \epsilon \left (\sum_{h:|h| < k \land H \not{\in} \cal{H'}} \frac{P(Z_{i,\cdot}|h)}{P(Z_{i,\cdot})} \right) =
$$ $$
\sum_{h:|h| < k \land h \in \cal{H'}} \frac{P(Z_{i,\cdot}|h)(\pi(h) - \epsilon)}{P(Z_{i,\cdot})}  +  \epsilon \left ( \sum_{h:|h| < k } \frac{P(Z_{i,\cdot}h)}{P(Z_{i,\cdot})} \right )
$$
Finally, since $P(Z_{i,\cdot}) = \sum_h P(Z_{i,\cdot}|h)\pi(h) \ge \sum_{h \in \cal{H'} } P(Z_{i,\cdot}|h)\pi(h)$, we obtain:

$$
fdr_k(Z_{i,\cdot})  \le \frac{\sum_{h:|h| < k \land h \in \cal{H'}} P(Z_{i,\cdot}|h)(\pi(h) - \epsilon) + \epsilon \sum_{h:|h| < k } P(Z_{i,\cdot}|h)}{\sum_{h \in \cal{H'} } P(Z_{i,\cdot}|h)\pi(h)}
$$\\
\end{proof}

The term above can be calculated in $O(m(n_H+k))$. First, the terms $\sum_{h:|h| < k \land h \in \cal{H'}} P(Z_{i,\cdot}|h)(\pi(h) - \epsilon) $ and $\sum_{h \in \cal{H'} } P(Z_{i,\cdot}|h)\pi(h)$ are calculated directly using the output of our EM-like algorithm in $O(mn_H)$. Then, $\epsilon \sum_{h:|h| < k } P(Z_{i,\cdot}|h)$ can be calculated using dynamic programming in a similar fashion to our algorithm for calculating $fdr_k$ under independence assumption. Thus, the total running time for all genes is $O(mn(n_H +k))$.

\subsection*{Replicability Across Independent Study Clusters}

\subsubsection*{Method overview}

In this section we apply the ideas from the previous sections to obtain an algorithm for calculating the $fdr$ under the assumption that the studies originate from independent clusters. We call this approach SCREEN (Scalable Cluster-based REplicability ENhancement). Briefly, our algorithm has three stages. First, we use the EM-process on each study pair to create a network of study correlations. We then cluster the network to obtain a set of study clusters that are likely to be independent, see \textbf{Materials and Methods} for the full description of this step. Second, we run the EM approach on each cluster separately. Finally, we merge the results from the different clusters using dynamic programming. Note that this algorithm is a heuristic as it uses EM within each cluster.

\subsubsection*{An algorithm for combining study clusters}

Assume for now that we are given a clustering of the studies into $M$ independent clusters $C_1,\cdots,C_M$. Thus:
$$
\pi(h) = \prod_{j=1}^M P(h_{C_j})
$$
where $h_{C_j}$ denotes the subvector of $h$ confined to the studies in the cluster $C_j$.
Let $Z_{i,C_j}$ be the z-scores of gene $i$ in the studies of cluster $C_j$, then $fdr_k$ has the following form:
$$
\sum_{h:|h|<k} \ \prod_{j=1}^M \frac{P(Z_{i,C_j}\big |h_{C_j})P(h_{C_j})}{P(Z_{i,C_j})}
$$

This form is a generalization of the formulation under indepepndence assumption, which implies that the dynamic programming approach can be used to merge data across clusters. We now describe the full method.

We apply our EM approach to each cluster separately, and calculate the probability that gene $i$ has exactly $k^*$ non-null realizations in cluster $C_j$: 
$$
V_{C_j}[i,k^*] = \sum_{h_{C_j}:|h_{C_j}|=k^*-1} P(h_{C_j}|Z_{i,C_j})
$$

Let $V[i,j,k^*]$ be the probability that gene $i$ has $k^*-1$ non-null realizations over clusters $1,\cdots,j$. Then:
\begin{eqnarray*}
V[i,1,k^*] &=& \begin{cases} 
   V_{C_1}[i,k^*] & \text{if } k^* < |C_1| \\
   0       & \text{otherwise}
  \end{cases} \\ 
V[i,j,k^*] &=&  \sum_{k'=1}^{min(k^*,|C_j|)} V_{C_j}[i,k'] V[i,j-1,k^*-k']\\
fdr_k (Z_{i,\cdot}) &=& \sum_{k^*=1}^{k-1} V[i,m,k^*]
\end{eqnarray*}
The table $V$ above has $M \times k$ entries for each gene $i$, and the update rule takes $O(k)$. Thus, given the EM results in each cluster, the running time of this algorithm is $O(k^2M)$ for each gene.

\subsection*{Experimental Results}

\subsubsection*{Simulations}

We first tested the performance of several methods (including SCREEN) in detecting genes that are non-null in multiple studies using simulated data.

\textit{Simulation overview:} In each of the scenarios below we first started by creating a pair of matrices, $P$, and $H^P$. The number of genes $n$ was 5000 and the number of studies $m$ varied across the simulation scenarios below (in all scenarios $m\ge 20$). $P$ is a matrix of p-values, and $H_{i,j}^P \in \{0,1\}$ denotes whether the p-value of gene $i$ in study $j$ is from the null group ($H_{i,j}^P=0$) or the non-null group ($H_{i,j}^P=1$). We first simulated $H^P$, and then simulated $P$ given the gene configurations in $H^P$ as follows. For all cells $i,j$ where $H_{i,j}^P=0$ the values were randomly selected from a uniform distribution. For each cell $i,j$ for which $H_{i,j}^P=1$ the p-value was drawn with probability $0.5$ from $\beta(1,x)$ (i.e., low p-values), and otherwise from $\beta(x,1)$ (i.e., high p-values). We tested $x=$10, 100, and 1000. In addition to the non-null distributions, the scenarios below also differ in other parameters in the creation of $H^P$: the number of non-nulls, and the correlation structure among the studies (i.e., the columns of $H^P$).

\textit{Compared methods:} We evaluated six different approaches for replicability, see Materials and Methods for details on 1-3. (1) Fisher: Fisher's meta-analysis for each gene with a BH correction, (2) Exp-count: an estimator for the expected number of non-nulls, (3) BH-count: the number of q-values of at most $0.1$ for each gene after applying BH in each study, (4) SCREEN-ind: our dynamic programming algorithm for $fdr_k$ under independence assumption, (5) repfdr-UB: our algorithm that computes an upper bound for $fdr_k$ (it can be viewed as an extension of repfdr that can handle many studies), and (6) SCREEN: our approach for replicability across study clusters. 

\textit{Performance evaluation:} We tested the ability of the methods to detect genes that are non-null in several studies (i.e., genes with at least two non-null realizations). Here, the true parameter of a gene was the number $k$ of 1 values in its row in $H^P$. For each $k$ between 2 and 5 we ran all algorithms above. For repfdr-UB we set the number of configurations $n_H$ to $512$. For methods that are based on calculating the local fdr we used a threshold of $0.2$ to select genes (for each $k$). For methods that are based on counting we used $k$ as a threshold. For Fisher's meta-analysis we used $q\le0.1$ as a threshold. For a given $k$ we compared the output of the algorithms to the set of genes with at least $k$ non-null realizations (given in $H^P$). We calculated two scores to quantify the performance: the Jaccard coefficient and the false discovery proportion (FDP; i.e., the proportion of erroneously declared non-nulls).

\subsubsection*{Scenario 1: independent studies}

Here, a random set of 300 genes were selected to be non-nulls independently in each study. In addition, we selected 50 genes to have non-null $\beta(1,x)$ p-values in 5 additional studies.

The performance of the algorithms for $x=1000$ using locfdr and normix were similar (\textbf{Figures  1A, 1B}). In addition, the results illustrate why Fisher's meta-analysis test is not suitable for our goals: the FDP was very high ($\ge 0.25$), even for $k=2$ or $k=3$. In addition the Jaccard scores were low ($\le 0.3$). These results indicate that such meta-analysis methods have low power in detecting genes with reapearring signal. All other methods performed much better. Specifically, for low $k$ values, SCREEN-ind, SCREEN and Exp-count were superior and performed similarly. For larger $k$ values, the performance of these methods was high, and BH-count achieved the top performance. For the harder case of $x=100$ (\textbf{Figure  1C}) the Jaccard scores were much lower for all methods. Except for repfdr-UB, and Fisher's method, all FDP scores were low. When $x=10$ the $FDP$ scores of repfdr-UB were high (e.g., $\ge 0.25$) for each $k$ as well as the FDP scores of $SCREEN$ for $k=2$ using the normix method. All other FDP scores were very close to zero, and when the FDP values of SCREEN were high, very few genes where reported($\le 4$), see \textbf{Supplementary Figure  1}.

\begin{figure}[!ht]
   \centering
   \includegraphics[width=150mm,height=160mm]{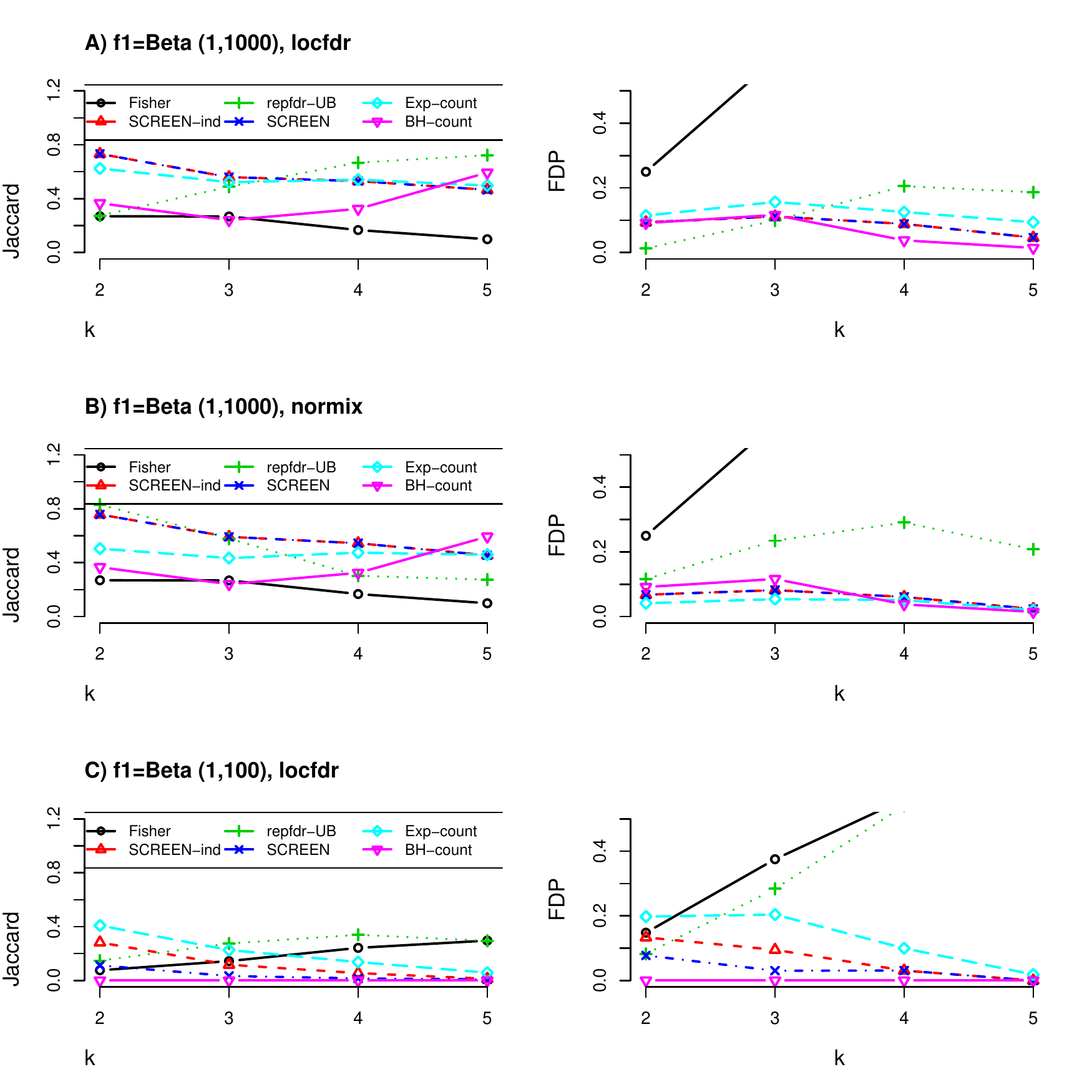}
    \caption{Simulation results: 20 independent studies (Scenario 1). A,B, and C represent different tests. A and B use the same non-null distribution in each study but a different method to learn the two-groups model. A and C use different non-null distributions in each study but the same method (locfdr) to learn the two-groups model. The left column shows the Jaccard scores and the right column shows the FDP scores. These scores are calculated by comparing the output gene set of each method for each $k$ to the set of genes for which the real number of non-nulls was at least $k$.}
    \label{fig:mesh1}
 \end{figure}

\subsubsection*{Scenario 2: dependent studies}

Here, the matrix $H^P$  was generated by first creating an auxiliary matrix $A$ of the same dimensions as $H^P$. The rows of $A$ were drawn independently from a multivariate normal distribution ${\cal N}(0, \Sigma_M)$, where $\Sigma_M$ specified a correlation structure of $M$ independent study clusters. Within clusters we set the correlation to $r=0.8$, or $r=0.4$, among all cluster studies. Finally, $H^P$ was created from $A$ by setting a threshold such that the expected number of non-nulls in each study was 300. That is, $H^P_{i,j}=1$ if and only if $A_{i,j} \ge \tau^j$, where $\tau^j$ is the $0.94$-th quantile of the normal distribution of column $j$ of $A$.
Given $H^P$, $P$ was created as in Scenario 1: null instances were drawn independently from a uniform distribution, whereas non-null instances were drawn from a distribution of lower p-values by setting $x=100$ (i.e., the non-nulls follow $\beta(1,100)$ or $\beta(100,1)$).

We tested two scenarios for each $r$: a single cluster (i.e., $M=1$) of 20 studies, or four clusters (i.e., $M=4$) of ten studies each. The results for $r=0.8$ are shown in \textbf{Figure 2}, and the results for $r=0.4$ are shown in \textbf{Supplementary Figure 2}. In terms of $FDP$, all algorithms except for Fisher's method and Exp-count performed well. In terms of the Jaccard score, the results were more mixed. First, for $r=0.8, M=1$ repfdr-UB and SCREEN were equivalent (as a single cluster was detected correctly) and reached the top performance ($\ge 0.8$). In all other scenarios the SCREEN algorithm achieved top, or nearly top, performance for large $k$ values ($k\ge4$), and $repfdr-UB$ had high FDP values for $r=0.4$. For lower $k$ values the SCREEN-ind approach and Exp-count achieved top performance with Jaccard $\ge 0.6$ in all tests. \textbf{Figure 3} shows the results of SCREEN for different $k$ values ($r=0.8$, $M=4$) compared to the real fdr values. The results for $k=4$ show that our estimations are highly correlated with the real values (\textbf{Figure 3A}), and that the estimated fdr values decrease with the real number of non-nulls (\textbf{Figure 3B}).  

In summary, our simulations show that out of the $fdr_k$-based methods, SCREEN-ind and SCREEN had low FDP values in all tests, while achieving high Jaccard performance: SCREEN-ind had a slight advantage in low $k$ values, whereas SCREEN had a slight advantage in higher $k$ values. The Exp-count method also performed well achieving high Jaccard scores, but had relatively high FDP values in some of the tests.

\begin{figure}[!ht]
   \centering
   \includegraphics[width=150mm,height=160mm]{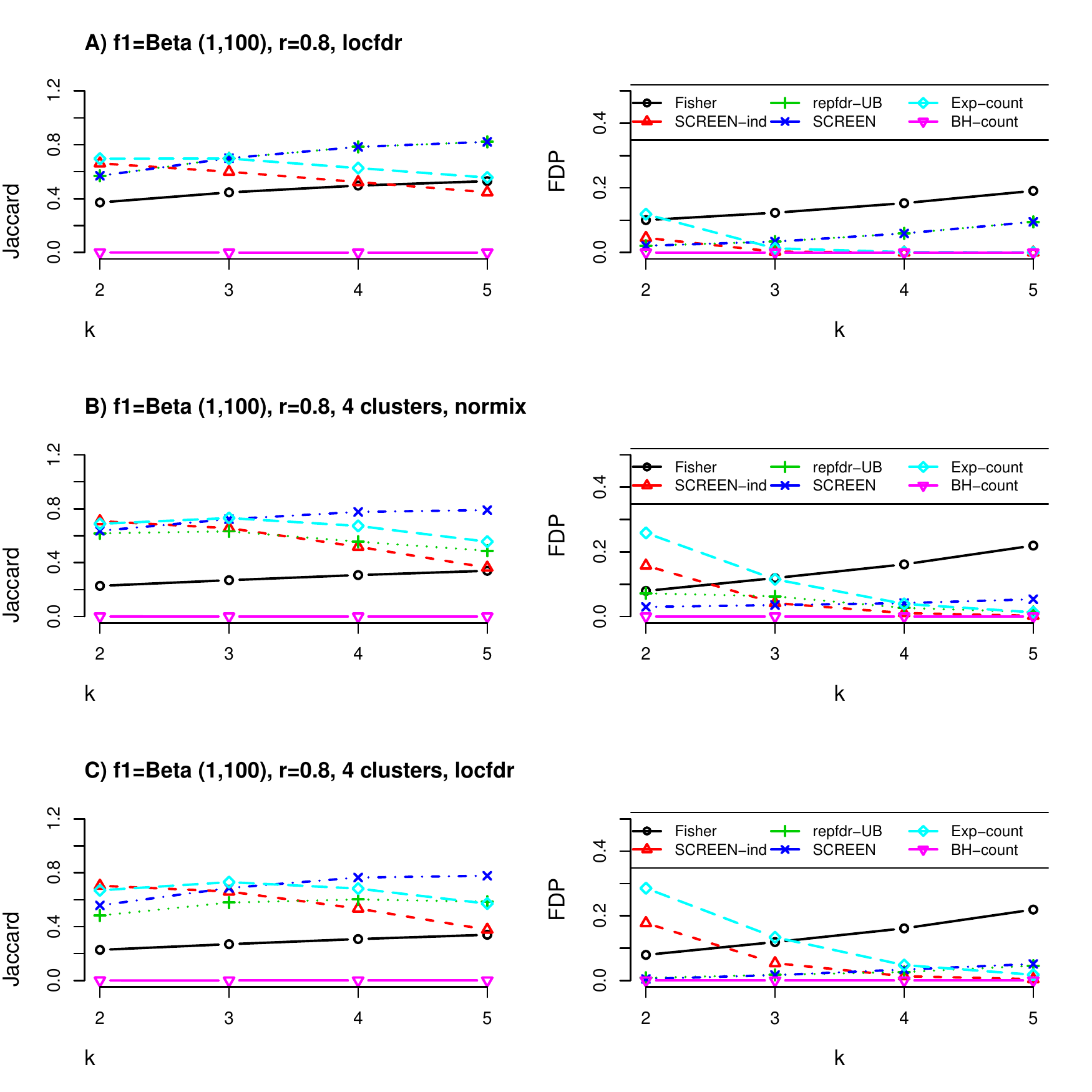}
    \caption{Simulation results: independent study clusters with high dependence within clusters (Scenario 2). A, B, and C represent different tests. Each test uses a different method used to learn the two-groups model in each study (A,C: locfdr; B:normix), or a different number of clusters (A:1, B,C:4). The left column shows the Jaccard scores and the right column shows the FDP scores. These scores are calculated by comparing the output gene set of each method for each $k$ to the set of genes for which the real number of non-nulls was at least $k$.}
    \label{fig:mesh1}
 \end{figure}

\begin{figure}[!ht]
   \centering
   \includegraphics[width=150mm,height=120mm]{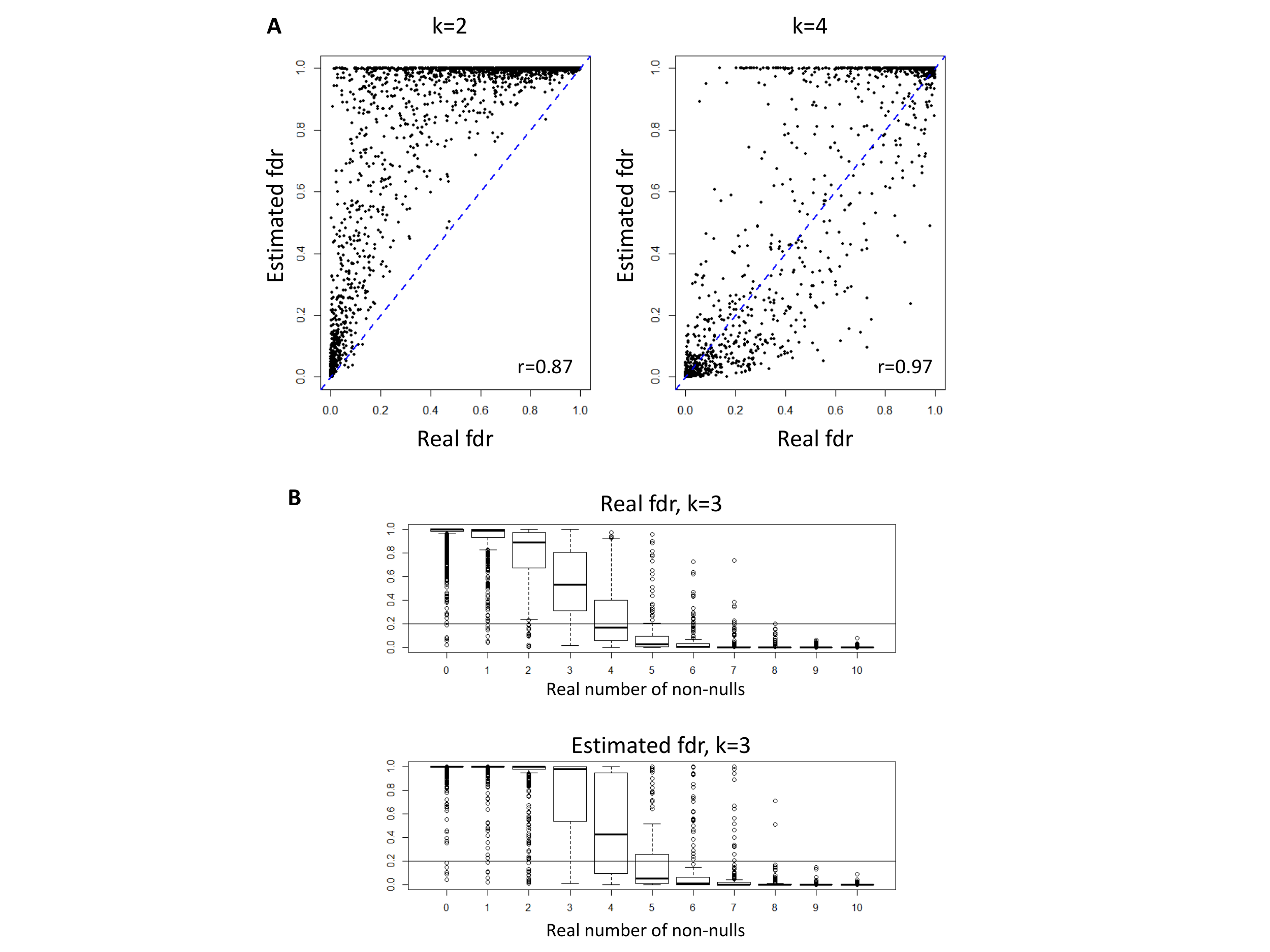}
    \caption{Simulation results: 4 clusters of 10 studies each. The figures show examples of SCREEN's estimations vs. real fdr values for different $k$ values. The number of study clusters is 4 and the correlation within the clusters is set using $r=0.8$. The non-null distribution within each study is $Beta(1,100)$. A) Real vs. estimated fdr values for $k=2$ and $k=4$. For $k=2$ the estimated $fdr$ values are higher, representing stringent FDR control. For $k=4$ the real and estimated values highly correlate. B) Boxplots of fdr distributions as a function of the real number of non-nulls. Up: real fdrs, down: estimated fdrs. Each boxplot represents a different set of genes whose real number of non-nulls is given in the x-axis label (except for 10, which means at least 10 non-nulls). The results show that most of SCREEN's errors are made for genes with $3$ or $4$ non-null realizations, and that the real fdr values will not necessarily capture these genes at $fdr_k \le 0.2$. On the other hand, a greater fdr threshold can be used (e.g., 0.4) to cover additional true negatives at the exspense of a few false positives (boxplots 1, 2, and ,3).}
    \label{fig:mesh1}
 \end{figure}

\subsubsection*{Cancer datasets}

We analyzed two real datasets. The first, which we call \textit{Cancer DEG}, is a collection of 29 gene expression studies that compared cancer to non-cancer tissues. The second, called \textit{HLA}, is from \cite{Shukla2015}. In this paper, Shukla et al. tested differential expression between cancer samples with and without somatic mutations in the HLA complex across 11 TCGA cancer subtypes.

In our simulations above the gene effects were sparse. That is, in all scenarios the non-null prior probability was relatively low. When we analyzed the real datasets we observed that while some studies were in line with these classic assumptions, many others were not, see \textbf{Supplementary Figure 3} and \textbf{Supplementary Figure 4}. To cope with these cases of dense effects we performed an extensive additional analysis on both simulated and real data, see \textbf{Supplementary Text}. Our main findings are as follows: (1) locfdr often fails to model these cases, (2) locfdr and normix with empirical null estimation overestimate the null prior probability, and (3) as reported in the previous section, using SCREEN with normix and a fixed theoretical null achieved very high Jaccard scores and low FDP on simulated data. We therefore use the latter approach to analyze the datasets in the subsequent sections.

\subsubsection*{The Cancer DEG dataset}

This dataset contains 29 microarray gene expression studies that compared cancer to non-cancer tissues. It was selected from our previously published compendium \cite{Amar2015} by taking all studies that had at least 10 cancer and 10 non-cancer samples (one study was excluded because its gene set was too small). For each dataset genes were assigned p-values for distinguishing between cancer and non-cancer classes using the GEO2R web tool of NCBI \cite{Barrett2013}. The resulting p-value matrix had 11540 rows (genes) and 29 columns, where the p-values were calculated using a two-tailed t-test for differential expression. \textbf{Supplementary Figure 5} shows the estimated pairwise correlations between studies. SCREEN identified eight clusters: a single large cluster of 19 studies and 7 clusters with one or more studies.

\textbf{Figure 4A} shows a comparison of SCREEN and SCREEN-ind in terms of the number of selected genes (at $0.2$ $fdr$), as a function of $k$, the minimum number of studies on which a gene must be detected. For low $k$ values (e.g., $k<10$) most genes had low $fdr$, suggesting that nearly all genes were differentially expressed in $k$ or more studies. For $k\le17$ values, SCREEN-ind reported more genes than SCREEN. However, SCREEN reported many more genes for higher $k$ values. For example, for $k=20$ SCREEN-ind detected 59 genes, whereas SCREEN detected 147 genes. In addition, for each $k$ we compared the output of each algorithm to Fisher's meta-analysis. Here, we used Spearman correlation to compare the gene ranking obtained by the methods. \textbf{Figure 4B} shows the results as a function of $k$. The correlation starts at high values that are close to 1 (for $k=2$) and decreases with $k$. \textbf{Figure 4C} depicts three examples of genes with different ranks: TOP2A, ATP6V1D, and GNPDA1. For each gene the plot shows the $-log_{10}$ p-value in each study, as well as the rank of the gene according to each of the methods. TOP2A was the top ranked gene in Fisher's meta-anlaysis, but had much lower ranks in SCREEN-ind and SCREEN. ATP6V1D and GNPDA1, which were the top two genes of SCREEN and SCREEN-ind, respectively, had much lower ranks in Fisher's meta-analysis. A comparison of the p-value patterns shows that ATP6V1D and GNPDA1 acheived higher rankings even though  TOP2A had more p-values that were extremely low (e.g., $<10^{-20}$). Thus, these examples show that our replicability analysis highlighted genes that were differential consistently across many studies, whereas meta-analysis (as expected) was more sensitive to extremely low p-values.

Our discoveries above were based on genes with consistently low p-values in cancer studies. However, the analysis did not use the direction of the differential expression. To shed more light on the directionality we analyzed the 147 genes detected by SCREEN for $k=20$. For each gene, we compared the times the reported t-statistic was negative and positive (corresponding to down- and up-regulation, respectively). Genes for which the number of negative values was at least thrice the number of positive values were denoted as down-regulated. Up-regulated genes were defined similarly as those for which the number of positive values was at least thrice the number of negative values. In total, there were 18 down-regulated and 99 up-regulated genes. The remaining 30 genes were denoted as \textit{mixed}. 

We used protein-protein interactions and pathway interactions as a gene network and plotted the subgraph induced by our 147 genes, see \textbf{Figure 4D}. Interestingly, the largest connected component was composed of 52 up-regulated genes and only a single down-regulated gene (KAT2B). Also, 88 out of 94 edges in the network connected up-regulated genes. Many of the genes in this "active" module were not ranked among the top 200 meta-analysis genes obtained using Fisher's method even though they are well known to play major roles in cancer formation and progression. For example. CDK1 is a master regulator of cell division, and CKAP5 is important for cytoplasmic microtubule elongation and is known to be over-expressed in colonic and hepatic tumors \cite{Charrasse1995,Enserink2010,Royle2012}. The most enriched GO term in the up-regulated gene set was mitotic cell cycle (38 genes, $q<10^{-28}$). The most enriched term in the active module was spindle organization (10 genes, $q=5.7\cdot10^{-11}$). Notably, several important genes, such as CDK1, MCM3, and MCM5, were not among the top 200 meta-analysis genes. In summary, our results show that SCREEN revealed a large gene set of consistently up-regulated genes in cancer that are highly relevant in function. On the other hand, while the results of the meta-analysis were also informative they are very sensitive to study-specific genes with extreme p-values and they do not tend to promote consistency. Thus, SCREEN was instrumental for separating the main up-regulated cancer genes that are consistent across most cancers, from other genes that manifest tissue-specific effects.

\begin{figure}[!ht]
   \centering
   \includegraphics[width=135mm,height=120mm]{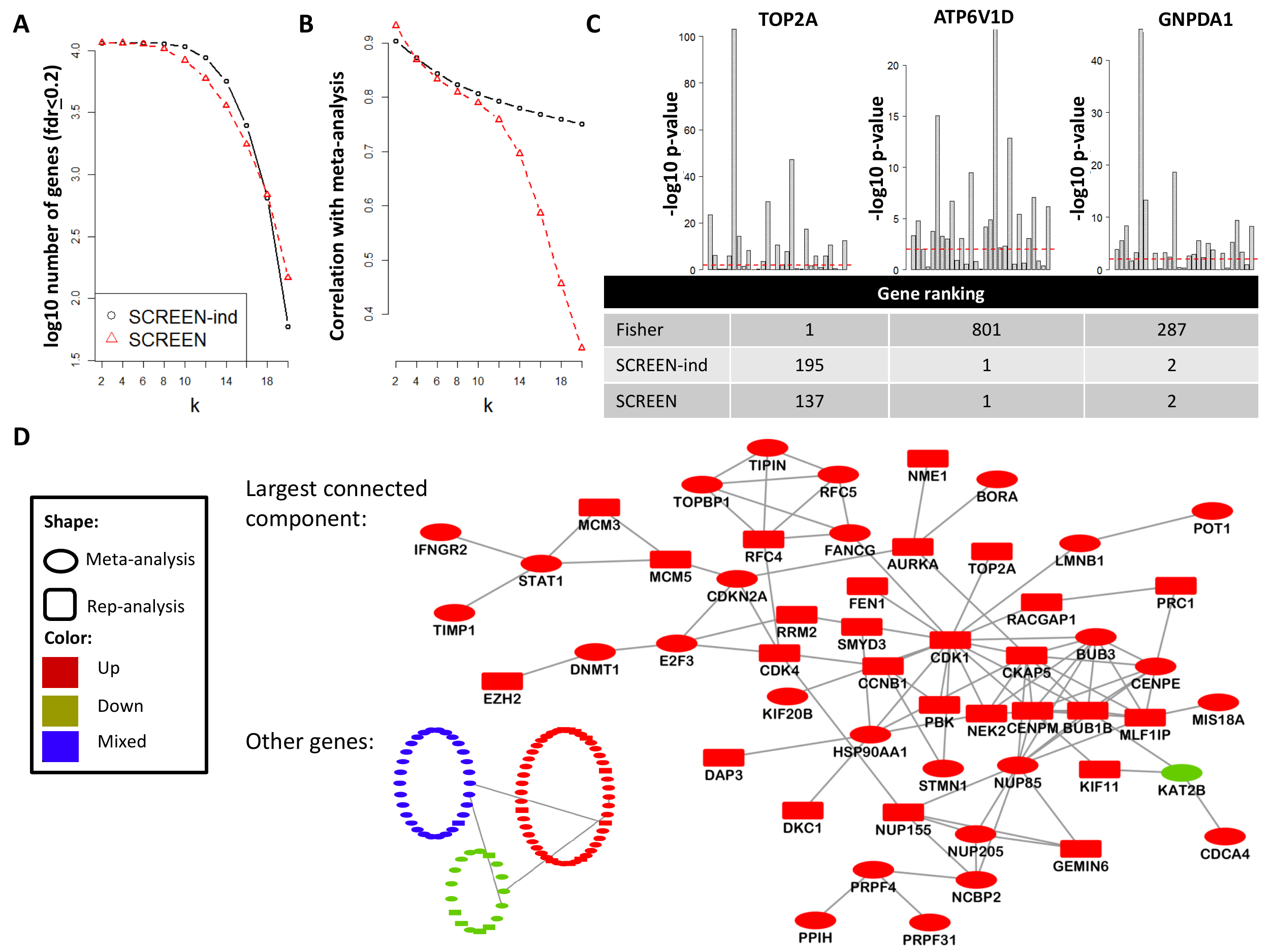}
    \caption{DEG dataset analysis. A) The number of reported genes at $0.2$ $fdr$ by SCREEN and SCREEN-ind as a function of $k$. B) The Spearman correlation between gene ranking of SCREEN and SCREEN-ind and of Fisher's meta-analysis as a function of $k$. C) The top ranked genes and their p-values. Top: the p-values of TOP2A, ATP6V1D, and GNPDA1 in each study. Bottom: the rank of these genes according to each of the methods (with k=20 for SCREEN and SCREEN-ind). ATP6V1D has a very low rank according to Fisher's meta-analysis even though it has consistently low p-values. D) Network analysis of the 147 genes reported by SCREEN with $k=20$. Nodes are genes, and edges are either protein-protein interactions or known pathway interactions. For each gene we calculated the number of up- and down-regulated t-statistics with a p-value $\le 0.01$. Genes for which the ratio between the up- and down events was $\ge 3$ ($\le 1/3$) were considered consistently up-regulated (down-regulated) in cancer (red and green nodes, 99 and 18 genes, respectively). All other genes were considered as mixed (blue nodes, 30 genes). Oval nodes represent genes ranked among the top 200 genes according to Fisher's meta-analysis (note that even at $10^{-5}$ Bonferroni correction, more than $10,000$ genes were selected in the meta-analysis, We therefore compared to the topmost genes, and chose the number 200 arbitrarily). Rectangular nodes are genes detected only by SCREEN.}
    \label{fig:mesh1}
 \end{figure}

\subsubsection*{The HLA dataset}

Shukla et al. \cite{Shukla2015} performed differential expression analysis between cancer samples with somatic mutations in the HLA complex and samples without such mutations across 11 different TCGA studies, each of a different cancer subtype. The p-value matrix, taken from \cite{Shukla2015}, had 18,128 rows (genes) and 11 columns. Similar to the Cancer DEG dataset, the p-values were based on a single-tail test for differential expression: p-values near zero represent up-regulation, and p-values near 1 represent down-regulation. Fisher's method was applied to assess the overall significance of a gene, and a total of 119 genes were selected using a p-value cutoff of $10^{-10}$. 

We applied both SCREEN and SCREEN-ind on these data. \textbf{Figure 5A} shows the number of selected genes (at $0.2$ $fdr$), as a function of $k$. Similarly to the DEG dataset, for low $k$ values ($k \le 3$) many genes (i.e. $\ge5000$ genes) had low $fdr$. However, there was a sharp decrease in the number of genes such that a few hundreds were found for $k=4$, and only a single gene was reported by SCREEN for $k\ge5$. When we compared the output of each algorithm to Fisher's meta-analysis (\textbf{Figure 5B}), we observed a similar trend of diminishing correlation with $k$. Here, the top ranked genes, TNNC2 and IFNG, had very high ranks in all methods, see \textbf{Figure 5C}.

For $k=4$ SCREEN reported many more genes than SCREEN-ind (405 vs. 135) genes, and both methods reported more genes that the original analysis (\textbf{Figure 5D}). We next performed functional analysis of the 452 genes detected by SCREEN. Here, most genes were consistently up- or down-regulated, and only 51 genes were mixed. Pathway enrichment analysis shows that our detected gene sets obtained similar results to the original publication and extended them, see \textbf{Supplementary Table 1}. Importantly, we recapitulated the main discovery of up-regulation of multiple immune related processes. Unlike the original publication, our analysis detected enrichment for cancer related pathways in the down-regulated gene sets of SCREEN, including Wnt signaling ($q=0.02$) and axon guidance ($q=0.04$). 

\textbf{Figure 5E} shows the largest connected component of the network induced by the up-regulated gene set. The network contains both genes reported in the original study and many newly reported genes, and the latter keep the component intact. Unlike the DEG dataset, the main connected component contains both up-, down-, and mixed-regulation patterns. However, the up-regulated genes form the backbone of the component. By focusing on the up-regulated genes we detected an active module of genes involved in activation of immune response, anti-tumor activity, and T-cell activation. Many high degree nodes in the module, such as  JAK2, were not reported in the original study. In summary, our results provide a comprehensive picture of the molecular response in patients with HLA mutations. SCREEN recapitulated the main findings obtained in the original study and revealed novel ones.

\begin{figure}[!ht]
   \centering
   \includegraphics[width=135mm,height=120mm]{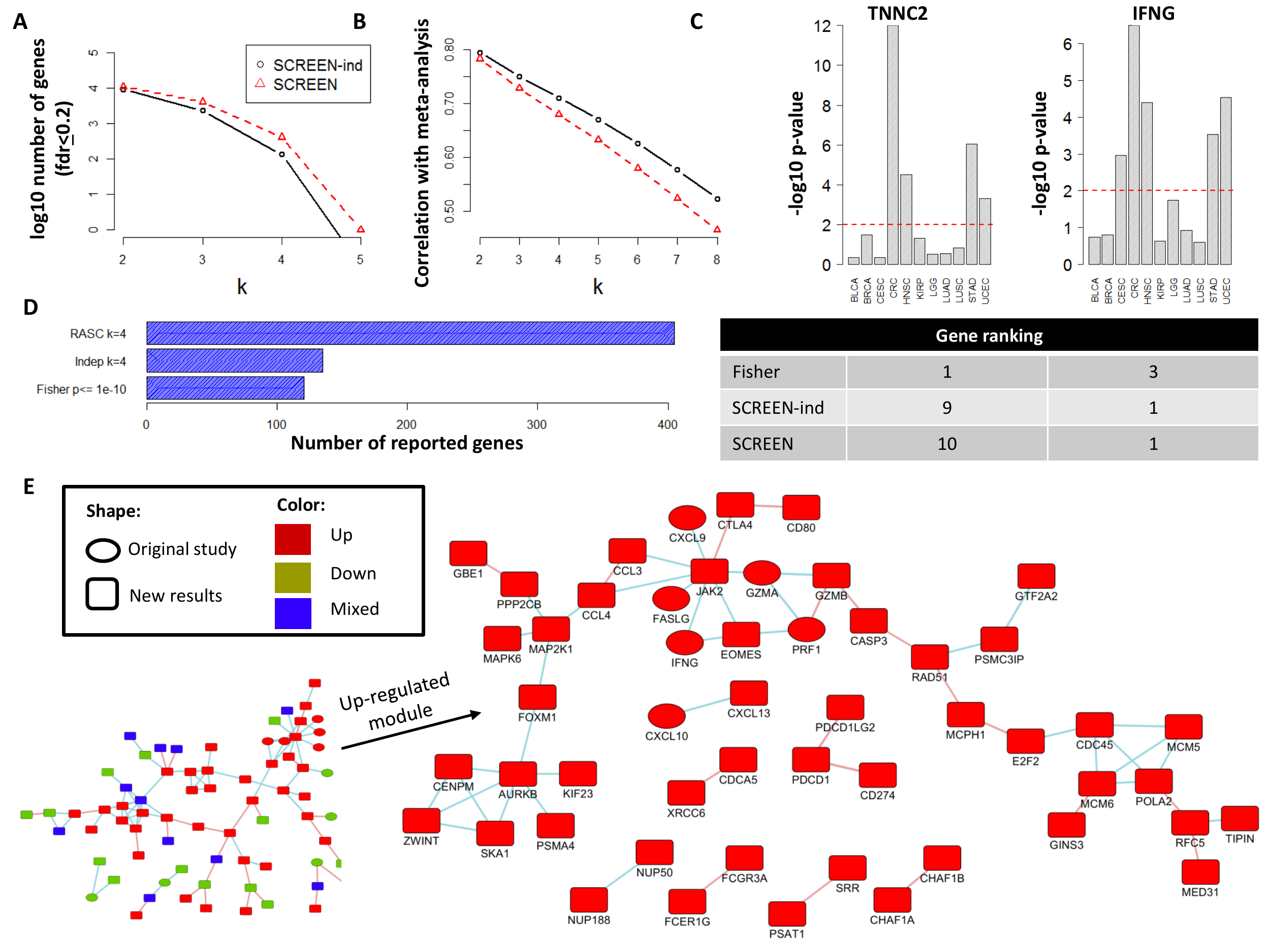}
    \caption{HLA dataset analysis. A) The number of reported genes at $0.2$ $fdr$ by SCREEN and SCREEN-ind as a function of $k$. B) The Spearman correlation between gene ranking of SCREEN and SCREEN-ind and of Fisher's meta-analysis as a function of $k$. C) The top ranked genes and their p-value in each study. Top: the p-values of TNNC2 and IFNG. Bottom: the rank of these genes according to each of the methods (with k=4 for SCREEN and SCREEN-ind). Overall, both genes are highly ranked by all methods. D) The number of genes reported by each method. E) The largest connected component produced by the 405 genes reported by SCREEN. Left: all genes, including up-, down-, and mixed-regulated genes. Right: focus on the up-regulated genes reveals an active module of immune response genes. Nodes are genes, and edges are either protein-protein interactions or known pathway interactions. As in the original study, this network suggests high activity of immune response. Two central genes in the immune response are INFG and JAK2. Our analysis detected both, whereas the original study did not detect JAK2. Moreover, the connectivity of the network is established by our newly detected genes.}
    \label{fig:mesh1}
 \end{figure}

\section*{Discussion}

In this paper we presented novel algorithms for empirical Bayes replicability analysis. We developed a new algorithm called SCREEN, which outperformed other approaches in many scenarios and consistently performed well in all simulations. SCREEN works in three stages. First, we cluster the studies based on their pairwise correlations, which are learned via EM. Second, we perform replicability analysis within each cluster using our restricted EM approach. This method extends previous studies by restricting the possible number of study configurations that are kept in memory. As a result, the method can analyze large study clusters. However, this is achieved at the expense of obtaining an upper bound for the $fdr$ instead of an exact estimation. Finally, the results of the replicability analyses of the clusters are merged using dynamic programming. For a given $k$, the output of SCREEN is the $fdr_k$ value for each gene, which can be used to detect genes that are non-null in at least $k$ studies. A possible extension of SCREEN is to have a more complex definition of replicability across study clusters. For example, a researcher may search for genes that are replicable across one or more study clusters, where a gene is replicable in a cluster only if it is non-null in at least some predefined percentage of the studies in that cluster. See \textbf{Supplementary Text} for an extended discussion on this topic.

We have shown that SCREEN performs well on various simulated scenarios, as well as on real datasets. Specifically, we analyzed two collections of cancer-related gene expression studies. In both cases the discovered gene sets highlighted active gene modules that are well connected in known gene interaction networks. Of note, the high connectivity of the modules is achieved thanks to genes that were not pointed out by standard meta-analysis. Thus, in this study, we both demonstrated replicability analysis as a standard tool for analyzing a large collection of studies, and provided novel algorithms that are accurate and scalable. 

Our study has some limitations that can be used as basis for future research. First, we assumed that genes are independent. While this assumption is usually made by state of the art methods, it is often incorrect. Second, while our algorithms report $fdr_k$ values of genes, we currently do not estimate their variance. Third, selection of $k$ was done manually on real datasets by looking at the number of reported genes (e.g., \textbf{Figure 4A}). Fourth, our restricted EM approach to analyze study clusters is a heuristic that only guarantees convergence into local optimum. Moreover, if the allowed maximum number of gene configurations is too low then our output is an upper bound for $fdr_k$ and not an exact estimate. Fifth, while our simple Exp-count approach to estimate the expected number of non-null realizations of a gene performed well in many simulated scenarios, it is only partially justified theoretically (see \textbf{Supplementary Text}). Finally, our methods rely on fixed estimates of the two-groups model of each study, and future methods could go a step further and estimate all parameters in a single flow (i.e., both the study parameters and the gene configuration probabilities).

\section*{Materials and Methods}

\subsection*{Learning a two-group model in each study}
We tried two different implementations of two-groups estimation algorithms. The first was $locfdr$ \cite{Efron2004}. This method implements two options to learn the empirical null: maximum likelihood and central matching. By default, we used the maximum likelihood estimator. However, in practice this algorithm might converge to a solution in which $ \hat \pi_0 > 1$. Whenever this occured, we tried the central matching approach instead. If the new estimator also had $ \hat \pi_0 > 1$ we used the theoretical null. 

The second approach was based on two previous methods: $Znormix$ \cite{McLachlan2006} and $fdrtool$ \cite{Strimmer2008a}. Znormix uses EM to learn a mixture of Gaussians, whereas $fdrtool$ assumes that the $null$ distribution is a half normal distribution. Here, we applied an EM approach to the absolute values of the z-scores. We extended these methods by learning a mixture of a half normal with $\sigma \ge 1$ for the null distribution, and a normal distribution with $\mu>0$ for the non-nulls. We call this approach $normix$. 

In practice, we discovered that our EM algorithm is sensitive to high values in the estimation of $f_1$. In addition, the methods above do not exploit additional information that could be obtained from the two-groups model: an estimation for the power of a study \cite{EfronBook2010}. That is, this is a measure of how separated the two groups are. 
In our analyses, we took a very stringent approach: in each study we multiply $f_1(z)$ by the estimated power of that study. The effect is a shrinkage in the $f_1(z)$ values that is proportional to the estimated quality of the study.

\subsection*{Clustering the studies}
Our algorithm above relies on a known partition of the studies into clusters. In this section we use an empirical Bayes approach to obtain the clusters. Our analysis has two main parts: learning a network, and clustering. 

First, we create a correlation network among the studies. For each study $i$, let $a_i=P(h_i=1)$ be the marginal non-null probability in that study. For studies $i,j$ let $a_{i,j} = P(h_i=1 \land h_j=1)$ be the shared non-null probability  of the two studies. We estimate these parameters as follows: $a_i$, and $a_j$ are taken from the two-groups model of each study, and $a_{i,j}$ is estimated by running our EM approach on the data of these two studies. The correlation of the studies is then estimated by:
$$r_{i,j} = \frac{a_{i,j}-a_i a_j}{\sqrt{a_i (1-a_i) a_j (1-a_j)}}$$
We obtain a robust estimation of $r_{i,j}$ by taking the mean of 100 bootstrap runs of the procedure above. That is, in each run we reestimate $a_{i,j}$ by running the EM on a bootstrap sample of the genes ($n/2$ genes out of $n$, sampled with replacement).

Next, we cluster the network using the infomap algorithm \cite{Rosvall2007}. Here, communities are detected using random walks in the underlying graph.  As the input for this algorithm is an unweighted network, we used a threshold of $0.1$ for the absolute correlation of study pairs to determine edge presence. This threshold is relatively low for general clustering tasks as it does not guarantee high homogeneity within clusters. However, it guarantees that the clusters discovered by SCREEN will be well-separated. In practice, our clustering approach found the correct clustering of studies in all simulations performed.

\subsection*{Other multi-study analyses}

In order to evaluate our $fdr$ approaches we compared them to several simple methods for multi-study analysis. Here we outline them briefly.

\subsubsection*{Fisher's meta-analysis}
We used Fisher's meta-analysis to merge the p-values of each gene into a single p-value. We then applied the BH FDR algorithm to account for multiple testing.
Note that Fisher's meta-analysis is not meant for replicability analysis and it does not take $k$ as input. Nevertheless, we added it to the comparison due to its use in recent publications (e.g., \cite{Schunkert2011,Kaever2014,Shukla2015}).

\subsubsection*{Counting-based BH analysis}
Here we run the BH multiple testing correction algorithm in each study separately. For each gene we count the number of q-values lower than some predefined threshold. In this study we used a threshold of $0.1$. We call this method \textit{BH-count}.

\subsubsection*{Counting-based tdr analysis}
In this analysis, for each gene, we use the \textit{local true discovery rates} (tdr) values obtained from the marginal two-groups model for each study. We then sum over these rates for the gene:
$$
\sum_{j=1}^m 1-P(h_{i,j}=0|Z_{i,j})
$$
This statistic can be interpreted as a biased estimator for the expected number of non-null realizations of gene $i$. See the \textbf{Supplementary Text} for more details. We call this method \textit{Exp-count}.

\subsection*{Enrichment and network analysis}
Network analysis and visualization was done in Cytoscape \cite{Shannon2003}. The GeneMANIA Cytoscape app \cite{Montojo2010,Vlasblom2014} was used to create the gene networks of the selected gene sets. GO enrichment analysis was performed using Expander \cite{Ulitsky2010}.

\section*{Acknowledgments}

This research was supported in part by the Israel Science Foundation as part of the ISF-NSFC joint program, and by the Israeli Center of Research Excellence (I-CORE), Gene Regulation in Complex Human Disease, Center No 41/11. Part of the work was done while DA and RS were visiting the Simons Institute for the Theory of Computing.


%
%
%



\pagebreak

\section*{Supplementary Text}

\subsection*{The expected number of non-nulls of a gene}
$$
E\left[ \sum_{j=1}^m H_{i,j}\big | Z_{i,\cdot} \right ] = \sum_{j=1}^m E\left[ H_{i,j}\big | Z_{i,\cdot} \right ] = \sum_{j=1}^m P\left[ H_{i,j}=1\big | Z_{i,\cdot} \right ]
$$
$$
= \sum_{j=1}^m \frac{ P(H_{i,j}=1)P(Z_{i,\cdot}|H_{i,j}=1) }{P(Z_{i,j})}  
 $$
 For each $j$ we now apply $P(Z_{i,\cdot}) = P(Z_{i,j})P(Z_{i,-j}|Z_{i,j})$ and $P(Z_{i,\cdot}|H_{i,j}=1) = P(Z_{i,j}|H_{i,j}=1)P(Z_{i,-j}|H_{i,j}=1,Z_{i,j})$, to partition $Z_{i,\cdot}$ into $Z_{i,j}$ and the remaining vector $Z_{i,-j}$. This produces the following form for our expectation:
 $$ E\left[ \sum_{j=1}^m H_{i,j}\big | Z_{i,\cdot} \right ] =  \sum_{j=1}^m \left ( \frac{ P(H_{i,j}=1)P(Z_{i,j}|H_{i,j}=1) }{P(Z_{i,j})} \right ) \left ( \frac{P(Z_{i,-j}| H_{i,j}=1,Z_{i,j})}{P(Z_{i,-j}| Z_{i,j})} \right )
$$
$$
 = \sum_{j=1}^m \left ( tdr_j(Z_{i,j}) \right ) \left ( \frac{P(Z_{i,-j}| H_{i,j}=1)}{P(Z_{i,-j}| Z_{i,j})} \right )
$$
The term above represents the expectation of non-nulls as a weighted sum over the true discovery rates of gene $i$ in each study $j$. The weight $\frac{P(Z_{i,-j}| H_{i,j}=1)}{P(Z_{i,-j}| Z_{i,j})}$ can be interpreted as a measure of discrepancy between observing that $H_{i,j}=1$ and observing $Z_{i,j}$. For example, if all studies are independent the weights are all $1$ and we get the sum of the tdr values. On the other hand, if all studies are highly correlated, the values in $Z_{i,-j}$ are very high (say $>5$), and $Z_{i,j}=0$, then $P(Z_{i,-j}| H_{i,j}=1)>P(Z_{i,-j}| Z_{i,j})$ and the weight of study $j$ will be $>1$, correcting upwards the low tdr value calculated for $Z_{i,j}=0$.

\subsection*{An extended discussion on analysis of study clusters}
Here we discuss a possible extension of our clustering analysis. If the researcher is looking for recurring genes across many clusters, but with a certain coverage in each cluster, then another definition of $fdr$ is required. More formally, we say that a gene is "interesting" in a cluster $C_j$ if it is non-null in at least $\lceil {\delta |C_j|} \rceil$ of the studies, where $\delta  \in (0,1]$. We say that a gene is "interesting" across clusters if it is interesting in at least $d$ clusters. Therefore, the rejection area for this analysis is determined by the two parameters $\delta$ and $d$. 

Let ${\cal{H}}_{\delta,d}$ be the group of configuration vectors in the null group:
$$
{\cal{H}}_{\delta,d} = \{ h:|\{ C_j; |h_{C_j}| < \delta |C_j| \}|< d \}
$$  
Then, the $fdr$ is defined as:
$$fdr_{\delta,d} (Z_{i,\cdot}) = \sum_{h \in {\cal{H}}_{\delta,d}} \frac{P(Z_{i,\cdot}|h)\prod_{j=1}^M P(h_{C_j})}{P(Z_{i,\cdot})} = \sum_{h \in {\cal{H}}_{\delta,d}} \ \ \prod_{j=1}^M \frac{P(Z_{i,C_j}|h_{C_j}) P(h_{C_j})}{P(Z_{i,C_j})} $$

Focusing on the first cluster, partition ${\cal{H}}_{\delta,d} $ into two groups: ${\cal{H}}_1^1 = \{h: h \in {\cal{H}}_{\delta,d} \land |h_{C_1}| < \delta |C_1| \}|\}$, and ${\cal{H}}_1^1 = \{h: h \in {\cal{H}}_{\delta,d} \land |h_{C_1}| \ge \delta |C_1| \}|\}$. Thus, we get:
$$
fdr_{\delta,d} (Z_{i,\cdot})  = \sum_{h \in {\cal{H}}_1^1} \frac{P(Z_{i,C_1}|h_{C_1}) P(h_{C_1})}{P(Z_{i,C_1})} fdr_{\delta,d} (Z_2,\cdots,Z_m) \ + \
$$
$$
\ \ \ \ \ \ \ \ \ \ \ \ \ \ \ \ \ \ \ \sum_{h \in {\cal{H}}_1^2} \frac{P(Z_{i,C_1}|h_{C_1}) P(h_{C_1})}{P(Z_{i,C_1})} fdr_{\delta,d-1} (Z_2,\cdots,Z_m)\\
$$
The formula above can be applied recursively for each cluster which suggests a dynamic programming calculation similar to our main algorithms. As in SCREEN, this calculation requires using the EM approach within each cluster.

This formulation requires a much more complex parameterization of the null group, such as determining the suitable $\delta$ for the application at hand. We expect that such formulation can be useful in future studies that will merge a large number of study clusters of similar sizes. 

\subsection*{Analysis of dense effects}
When plotting the p-values for individual studies in the cancer data that we analyzed (\textbf{Supplementary Figure 3} and \textbf{Supplementary Figure 4}) many studies seem to reflect very high proportions of non-null realizations. In order to evaluate the performance of locfdr and normix in this situation we performed an exploratory analysis. 

We analyzed the GSE10072 study, which shows a very uneven p-value distribution  (\textbf{Supplementary Text Figure A}). Here, using either normix or locfdr with estimation of the empirical null distribution resulted in very high $\pi_0$ estimates that seem to shrink the non-null group (i.e., $\pi_0 \ge 0.95$). On the other hand, when we used the theoretical null these estimations decreased substantially. As most of the density of the p-value distribution is concentrated at bins close to 0 or 1, the theoretical null estimation seems more suitable, see \textbf{Supplementary Text Figure A; C,D}. When we tried to use locfdr with theoretical null for all 29 studies of the DEG dataset, the algorithm failed to produce any output in 19 cases, even after we modified the default parameters, such as the number of degrees of freedom allowed, or the curve fitting method. For the reasons above we chose to use the normix approach with a fixed theoretical null for the analyses of the real datasets.

\begin{minipage}{\linewidth}
\makebox[\linewidth]{
 \includegraphics[width=150mm,height=100mm]{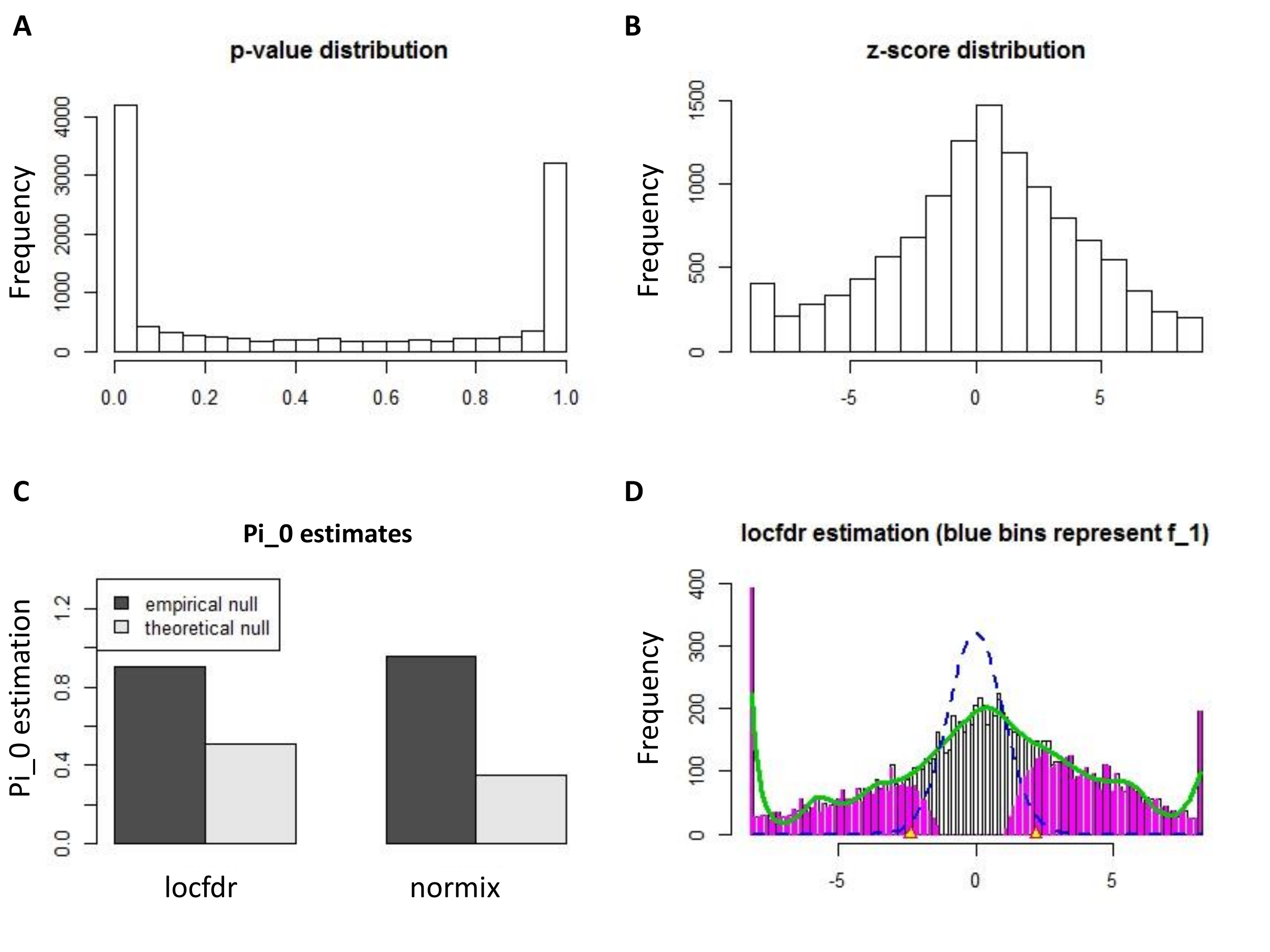}}
\textbf{Supplementary Text Figure A} Analysis of the GSE10072 dataset. A) histogram of the gene p-values. B) histogram of the gene z-scores. C) $\pi_0$ estimates using the two groups estimation algorithms with and without using the theoretical null. Empirical null estimation shrinks the non-null probability, whereas using the theoretical null captures the high percentage of non-null realizations. D) Standard plot of the locfdr estimation using the theoretical null. The plot shows the z-scores histogram with blue bins representing $f_1$. The dotted blue line is the standard normal distribution. The green line represents the estimated $f$ distribution.
\end{minipage}\\ \\

We also performed an additional simulation study in order to explore the scenario of dense effects. Using the notation from the simulations in the main text, we simulated $n=5000$ genes over $m=30$ studies. Studies 1-10 had no non-null realizations and studies 11-30 had 3000 ($60 \% $ of the genes). Studies 11-20 were all independent, studies 21-30 were all from a single cluster with $r=0.8$. To mimic the dense effects observed in some of the real datasets, the non-null realizations were obtained by sampling z-scores from a mixture of two normal distributions: $N(3,3)$ (corresponding to very low p-values), and $N(-3,3)$ (corresponding to very high p-values close to $1$). This simulates a case in which the variance of $f_1$ results from both standard noise levels (similar to the null distribution) and high noise levels of the gene effects. The results of the six tested algorithms on these data are shown in \textbf{Supplementary Text Figure B}. As in the GSE10072 dataset, locfdr with theoretical null did not report any output in many cases and it was therefore ommitted. The figure shows that using normix with a fixed theoretical null was {\color{red} markedly} better than all other approaches (e.g., Jaccard $0.78$ with theoretical null vs. $0.58$ with empirical null for SCREEN with k=7). As expected, as the null and non-null distributions are well separated in each study, most methods performed well and had high Jaccard scores and low FDP. SCREEN, Exp-count, Fisher, and SCREEN-ind all had high Jaccard scores with mild differences.

\begin{minipage}{\linewidth}
\makebox[\linewidth]{
 \includegraphics[width=160mm,height=140mm]{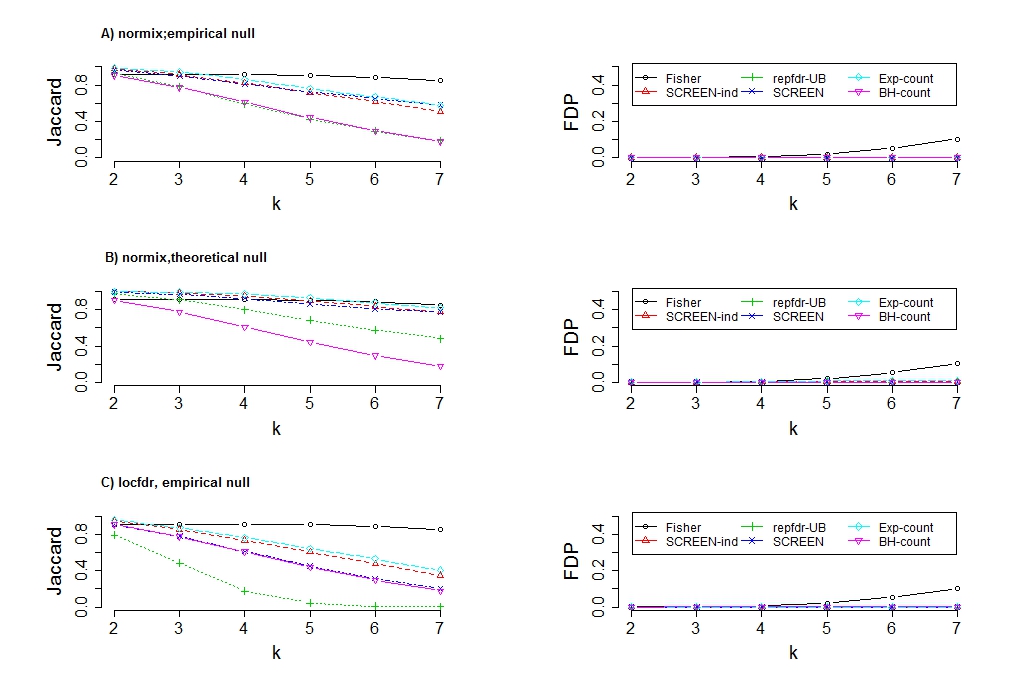}}
\textbf{Supplementary Text Figure B} Simulation results: a total of 30 studies, 20 of which (10 are in a single clusters and the rest are all independent) are with dense effects. The left column shows the Jaccard scores and the right column shows the FDP scores.
\end{minipage}

\pagebreak
\section*{Supplementary Tables}

\textbf{Supplementary Table 1} The studies of the cancer DEG dataset
\begin{center}
    \begin{tabular}{ | l | l | l | p{8cm} |}
    \hline
    GEO ID [ref] & \#cancer & \#control & description \\ \hline
	GSE2549 \cite{Sugarbaker2005} & 42 & 11 & Malignant pleural mesothelioma \\ \hline
	GSE2719 \cite{Detwiller2005} & 39 & 15 & Soft tissue sarcoma (multiple tissues) \\ \hline
	GSE4107 \cite{Hong2007} & 10 & 10 & Colonic mucosa from colorectal cancer patients and controls \\ \hline
	GSE4115 \cite{Spira2007} & 79 & 73 & Bronchial Epithelium from smokers with or without lung cancer \\ \hline
	GSE4290 \cite{Sun2006} & 157 & 23 & Different types of glioma vs healthy controls \\ \hline
	GSE5764 \cite{Turashvili2007} & 10 & 20 & Invasive lobular and ductal breast carcinomas \\ \hline
	GSE6344 \cite{Gumz2007} & 12 & 12 & Kidney cancer \\ \hline
	GSE6691 \cite{Gutierrez2007} & 43 & 13 & B cells from several hematopoeitic cancer types vs healthy B-cells or plasma cells \\ \hline
	GSE7803 \cite{Zhai2007} & 31 & 10 & Cervical squamous cell carcinomas \\ \hline
	GSE8671 \cite{Sabates-Bellver2007} & 32 & 32 & Colorectal adenoma \\ \hline
	GSE9476 \cite{Stirewalt2008} & 26 & 38 & Acute myeloid leukemia \\ \hline
	GSE9574 \cite{Tripathi2008} & 15 & 14 & Breast epithilium, breast cancer \\ \hline
	GSE9750 \cite{Scotto2008} & 33 & 33 & Cervical cancer \\ \hline
	GSE10072 \cite{Landi2008} & 58 & 49 & Lung adenocarcinoma (smoking individuals) \\ \hline
	GSE12452 \cite{Sengupta2006} & 31 & 10 & Nasopharyngeal carcinoma \\ \hline
	GSE12453 \cite{Brune2008} & 42 & 20 & Lymphocyte-predominant Hodgkin lymphoma \\ \hline
	GSE14245 \cite{Zhang2010} & 12 & 12 & Saliva, Pancreatic cancer vs controls \\ \hline
	GSE14407 \cite{Bowen2009} & 12 & 12 & Ovarian cancer \\ \hline
	GSE14520 \cite{Roessler2010} & 22 & 24 & Hepatocellular carcinoma \\ \hline
	GSE19804 \cite{Lu2010} & 60 & 60 & Non-smoking lung cancer \\ \hline
	GSE20189 \cite{Rotunno2011} & 73 & 80 & Peripheral whole blood, lung cancer \\ \hline
	GSE20347 \cite{Hu2010} & 17 & 17 & Esophageal squamous cell carcinoma \\ \hline
	GSE20437 \cite{Graham2010} & 18 & 24 & Epithelium from breast cancer patients and cancer-free prophylactic mastectomy patients \\ \hline
	GSE22529 \cite{Gutierrez2010} & 41 & 11 & Chronic lymphocytic leukemia \\ \hline
	GSE26566 \cite{Andersen2012} & 104 & 59 & Cholangiocarcinoma	\\ \hline
	GSE26910 \cite{Planche2011} & 12 & 12 & Stroma, breast and prostate cancers vs. controls \\ \hline
	GSE27562 \cite{LaBreche2011} & 57 & 105 & PBMCs from breast cancer patients and controls \\ \hline
	GSE28735 \cite{Zhang2013} & 45 & 45 & Pancreatic cancer \\ \hline
	GSE32665 \cite{Kim2013} & 87 & 92 & Lung adenocarcinoma \\ \hline
    \hline
    \end{tabular}
\end{center}

\bibliographystyle{ieeetr}

\pagebreak
\section*{Supplementary Figures}

\begin{minipage}{\linewidth}
\makebox[\linewidth]{
 \includegraphics[width=150mm,height=150mm]{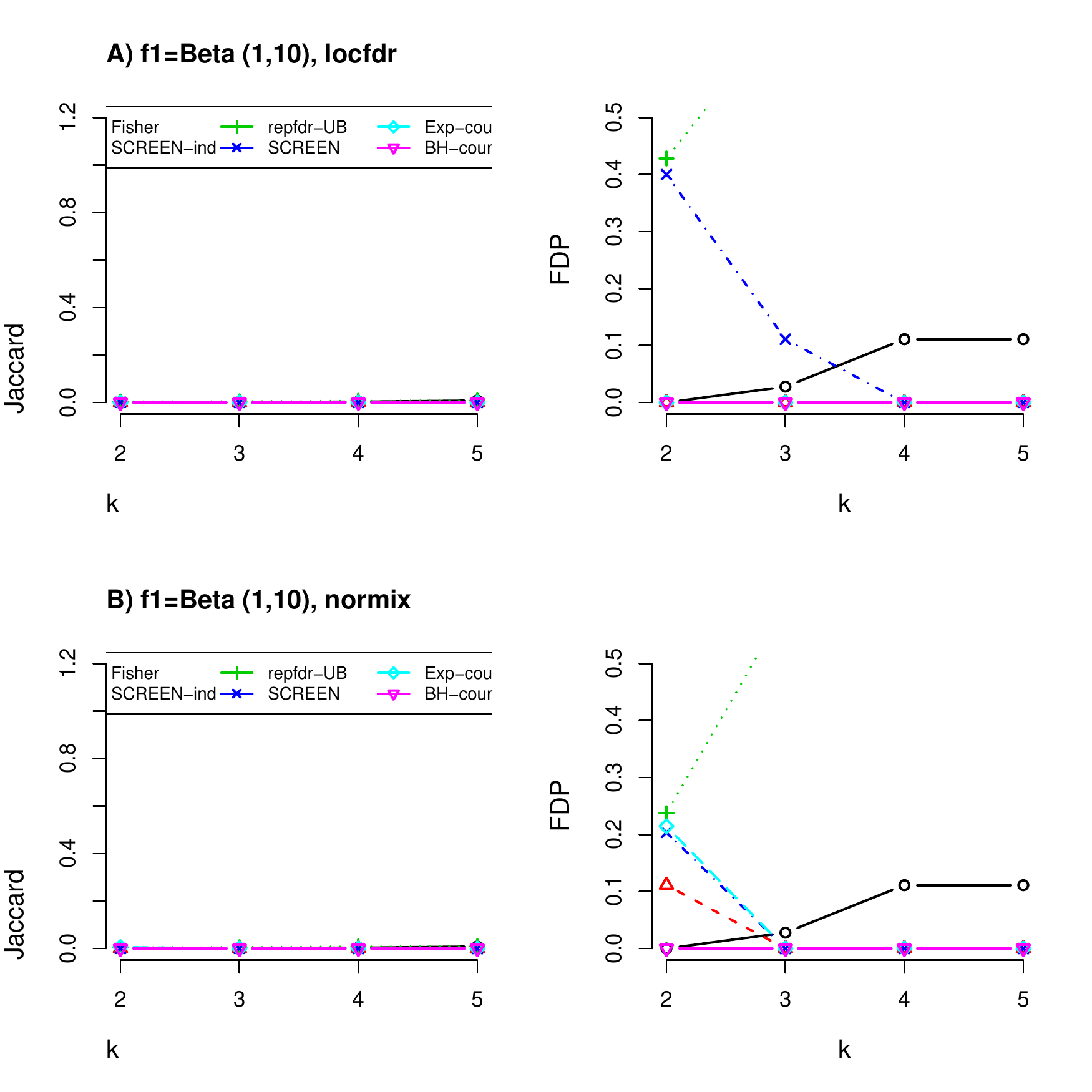}}
\textbf{Supplementary Figure 1} Simulation results: 20 independent studies with $\beta(1,10)$ non-null p-values. The left column shows the Jaccard scores and the right column shows the FDP scores. These scores are calculated by comparing the output gene set of each method for each $k$ to the set of genes for which the real number of non-nulls was at least $k$. The Jaccard scores here are always very low, and the FDP scores of SCREEN and repfdr-UB might be high. However, when the FDP scores are high, very few genes are reported by SCREEN ($\le 4$), whereas EM-UB might report $\ge 10$ genes.
\end{minipage}\\ \\

\pagebreak
\begin{minipage}{\linewidth}
\makebox[\linewidth]{
 \includegraphics[width=150mm,height=150mm]{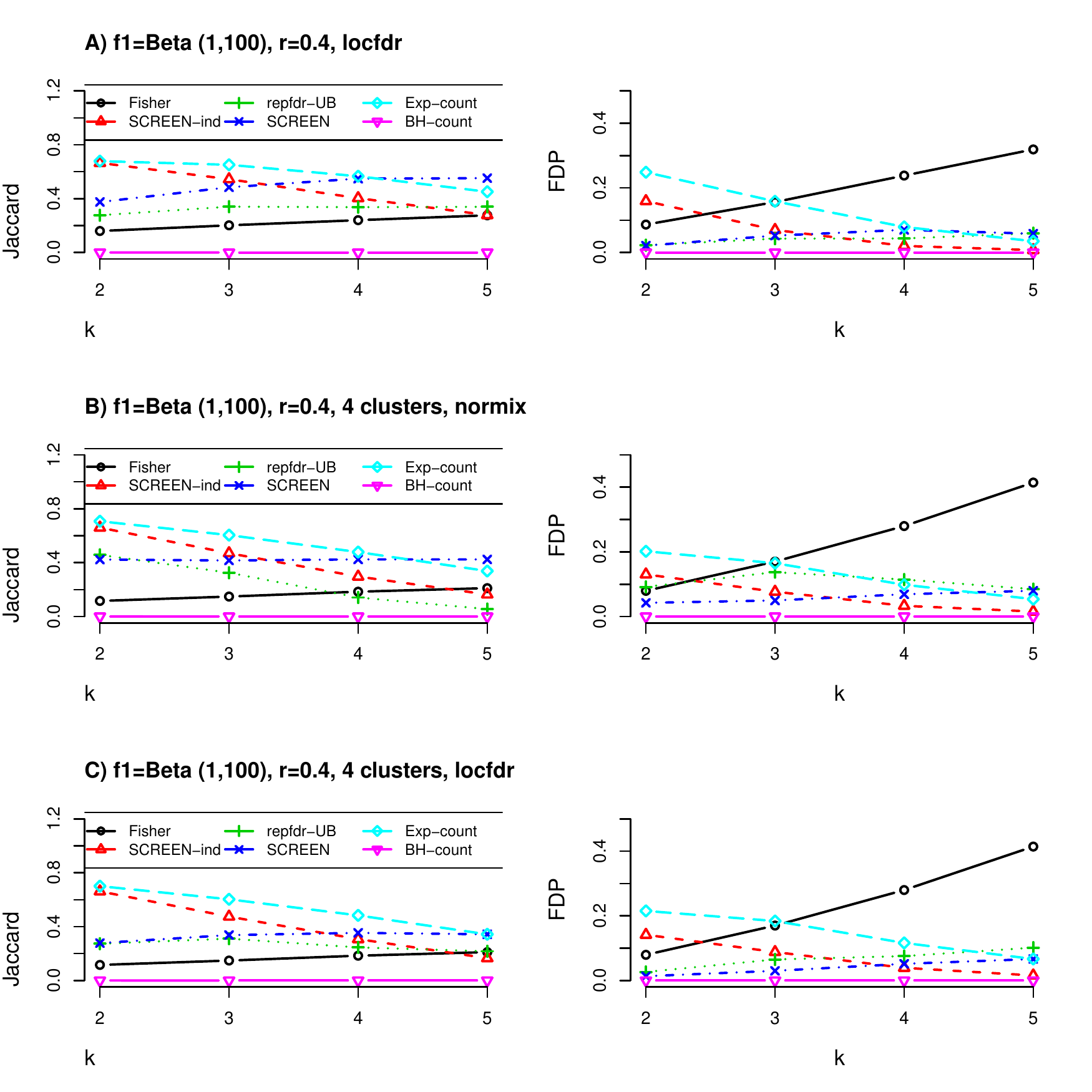}}
\textbf{Supplementary Figure 2} Simulation results: independent study clusters with mild dependence within clusters (Case 2 with $r=0.4$). Each row represents a different test. Each test has either a different non-null distribution in each study, a different method used to learn the two-groups model in each study, or a different number of clusters (A:1, B,C:4). The left column shows the Jaccard scores and the right column shows the FDP scores. These scores are calculated by comparing the output gene set of each method for each $k$ to the set of genes for which the real number of non-nulls was at least $k$.
\end{minipage}\\ \\

\pagebreak
\begin{minipage}{\linewidth}
\makebox[\linewidth]{
 \includegraphics[width=150mm,height=150mm]{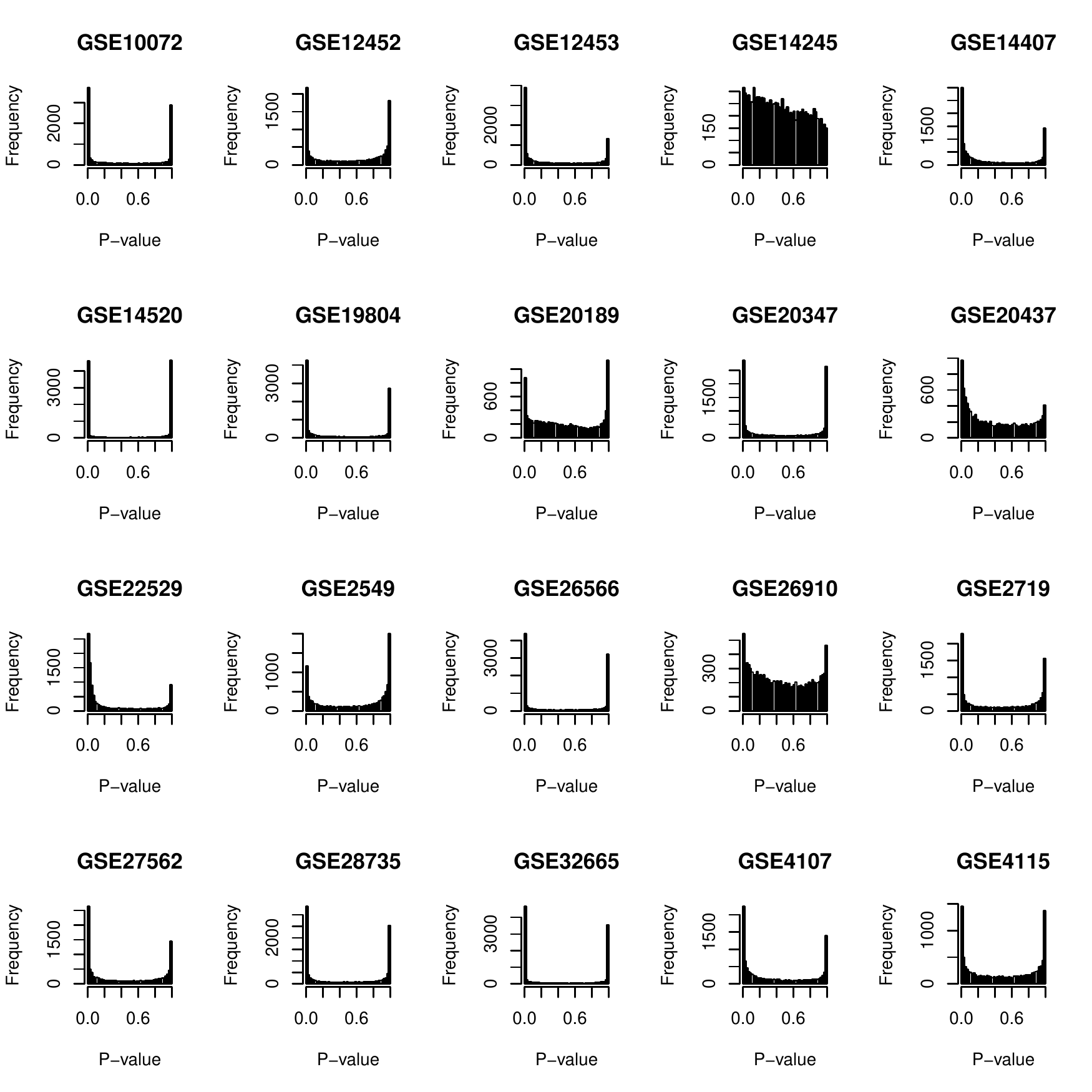}}
\textbf{Supplementary Figure 3} P-value histograms of 20 studies from the DEG dataset.
\end{minipage}\\ \\

\pagebreak
\begin{minipage}{\linewidth}
\makebox[\linewidth]{
 \includegraphics[width=150mm,height=150mm]{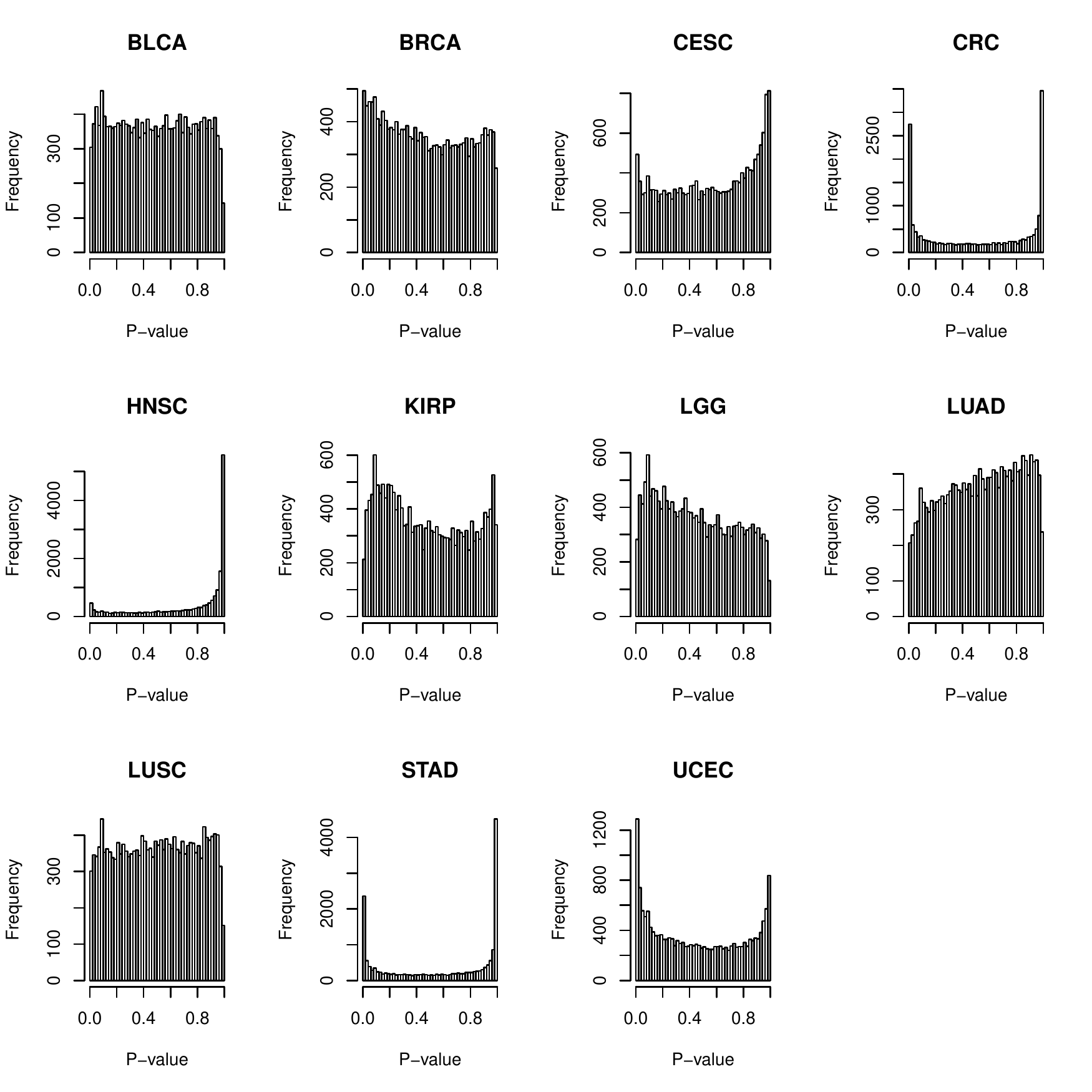}}
\textbf{Supplementary Figure 4} P-value histograms of all studies in the HLA dataset.
\end{minipage}\\ \\

\end{document}